%% file: scheduling.tex
\documentclass[lettersize,journal]{IEEEtran}
\usepackage{amsmath,amsfonts}
\usepackage{algorithmic}
\usepackage{algorithm}
\usepackage{array}
\usepackage{textcomp}
\usepackage{stfloats}
\usepackage{url}
\usepackage{verbatim}
\usepackage{graphicx}
\usepackage{cite}
\usepackage{amssymb}

\usepackage{caption}
\usepackage{subcaption}
\usepackage{xcolor}
\usepackage{graphicx}

\include{macros_reduced}

\begin{document}
	
		\title{Cellular, Cell-less, and Everything in Between: A Unified \\ Framework for Utility Region Analysis in Wireless Networks}
	
	\author{R.~L.~G.~Cavalcante${}^\dagger$,~\IEEEmembership{Senior Member,~IEEE}, T.~Piotrowski${}^\ddagger$, ~\IEEEmembership{Member,~IEEE}, S.~Stanczak${}^{\dagger\diamondsuit}$, \IEEEmembership{Senior Member,~IEEE}\\
		{\footnotesize ${\dagger}$ Fraunhofer Heinrich Hertz Institute, Germany \quad ${}^\ddagger$~Nicolaus Copernicus University, Poland \quad ${\diamondsuit}$~Technical University of Berlin, Germany \vspace{-0.5cm}}
		\thanks{The first author gratefully acknowledges Dr. Hiroki Kuroda for his valuable
			suggestions that improved the clarity of the proofs, and Dr. Lorenzo Miretti
			for insightful discussions and for sharing the simulation code used in this
			study. The first and third authors acknowledge the financial support from the Federal Ministry of Research, Technology and Space (BMFTR) in Germany, project xG-RIC (grants 16KIS2429K and 16KIS243), and from the 6G-MIRAI project, which has received funding from the Smart Networks and Services Joint Undertaking (SNS JU) under the European Union’s Horizon Europe research and innovation program under Grant Agreement No 10119236. Views and opinions expressed are, however, those of the author(s) only, and they do not necessarily reflect those of the European Union or the SNS JU (granting authority). Neither the European Union nor the granting authority can be held responsible for them.}
	}

	\maketitle

	\begin{abstract}
		We introduce a unified framework for analyzing utility regions of wireless networks, with a focus on signal-to-interference-plus-noise-ratio (SINR) and achievable rate regions. The framework provides valuable insights into interference patterns of modern network architectures, including extremely large MIMO and cell-less networks. A central contribution is a simple characterization of feasible utility regions using the concept of spectral radius of nonlinear mappings. This characterization provides a powerful mathematical tool for wireless system design and analysis. For example, it allows us to generalize existing characterizations of the weak Pareto boundary using compact notation. It also allows us to derive tractable sufficient conditions for the identification of convex utility regions. This property is particularly important because, on the weak Pareto boundary, it guarantees that time sharing (or user grouping) cannot simultaneously improve the utilities of all users. Beyond geometrical insights, these sufficient conditions have two key implications. First, they identify a family of (weighted) sum-rate maximization problems that are inherently convex, thus paving the way for the development of efficient, provably optimal solvers for this family. Second, they provide justification for formulating sum-rate maximization problems directly in terms of achievable rates, rather than SINR levels. Our theoretical insights also motivate an alternative to the concept of favorable propagation in the massive MIMO literature -- one that explicitly accounts for self-interference and the beamforming strategy.
	\end{abstract}
	 

	\section{Introduction}
	In wireless networks with users competing for shared resources, system designers must balance fairness with overall network efficiency. To this end, the performance of each user is evaluated via a utility function,  and this function guides resource allocation toward achieving this balance. Two common design criteria for selecting operating points are (weighted) max-min fairness \cite{nuzman07,cavalcante2019,tan2014wireless,miretti2022closed,chafaa2025}\cite[Ch.~3, Ch.~7]{demir2021}\cite[Ch.5.3.2]{marzetta16}\cite[Theorem 7.1]{massivemimobook} and (weighted) sum  utility maximization \cite{tan2011maximizing,tan2011nonnegative,zheng2013maximizing,cheng2016optimal}, and these criteria often use the signal-to-interference-plus-noise ratio (SINR) or inner bounds of the Shannon capacity as the utility function. 
	
	One of the main advantages of the max-min fairness criterion  is that it often leads to optimization problems that, although nonconvex in general, can be solved with efficient fixed point algorithms \cite{nuzman07}. In addition, this criterion is particularly suitable for networks designed to ensure uniform quality of service across all users, as in cell-less or user-centric networks. However, user utilities can often be further improved through time sharing (link scheduling/user grouping) \cite[Sect.~5.2.1]{slawomir09} because, for mathematical tractability, the utilities are typically based on capacity-region bounds or surrogate expressions based on the SINR, making the achievable utility region potentially nonconvex. While time sharing allows for greater flexibility in improving utilities, it also introduces a challenging combinatorial aspect to resource allocation. This aspect undermines one of the primary benefits of the max-min fairness criterion without time sharing: the availability of efficient algorithms able to compute global optima. In contrast, operating points that maximize the sum utility (typically the achievable rates) cannot be improved via time sharing. However, sum-rate maximization is NP-hard in general \cite{luo2008dynamic}, so it is often addressed with fast heuristics without any guarantees of obtaining global optima \cite{weeraddana2012,miretti2024sum}. 
	
	As discussed in the next subsection, the above challenges are eliminated if the feasible utility region is convex. In this scenario, time sharing is unnecessary in max-min utility optimization problems, and fast iterative algorithms that provably converge to global optima can potentially be devised for sum-rate maximization problems. However, to date, existing methods for identifying convex utility regions or tractable sum-rate maximization problems rely on assumptions that are difficult to satisfy in practice. They include equal interference-plus-noise for all users in the network \cite[Sect.~V.A]{cheng2016optimal}, symmetric interference patterns \cite[Ch.~5.4.4]{slawomir09}, no interference, or low transmit power \cite{zheng2013maximizing}\cite[Sect.~IV.B]{tan2011maximizing}. Consequently, there is a
	need to identify weaker, more practical assumptions.
	
	\vspace{-.3cm}
	
	\subsection{Time sharing in resource allocation}
	\label{sect.overview}

	To illustrate how the geometry of the utility region governs the need for time sharing in resource allocation, consider the uplink of a network with $N\in\Natural$ users represented by the set $\mathcal{N}=\{1,\ldots,N\}$, and let $\signal{p}=(p_1,\ldots,p_N)\in\real_{+}^N:=[0,~\infty[^N$ denote the transmit power vector, where $p_n$ represents the transmit power of user $n$. For each user $n\in\mathcal{N}$, we assign a utility function $U_n:\real_+^N\to\real_+$ that maps a given power vector $\signal{p}\in\real_+^N$ to the quality-of-service experienced by the user. 
	\footnote{These utilities take only the power vector as an argument, but other parameters to be optimized, such as the choice of beamformers and the user-base station assignment, can be implicitly considered in the definition of $U_n$. For concrete examples, see Sect.~\ref{sect.model} and \cite{boche2010unifying,boche2008,martin11,piotrowski2022,cavalcante2023,miretti2025two}.} In many wireless problems, which include cell-less systems, massive MIMO, and traditional cellular systems as particular instances,  the utility functions, such as the SINR or the achievable rates, can be unified as follows \cite{boche2010unifying,cavalcante2023}: 
	\begin{align}
		\label{eq.utility}
		(\forall n\in\mathcal{N})(\forall\signal{p}\in\real_+^N)~U_n(\signal{p})=p_n/t_n(\signal{p}),
	\end{align}
	where $(t_n:\real_+^N\to\real_{++})_{n\in\mathcal{N}}$ are standard interference functions, in the sense defined in \cite{yates95}, and $\real_{++}:=]0,\infty[$.
	In real systems, the power vector is typically constrained to a set of the type $\mathcal{P}:=\{\signal{p}\in\real_+^N\mid \|\signal{p}\|\le p_\mathrm{max}\}$, where $\|\cdot\|$ is a monotone norm (as defined in Sect.~\ref{sect.notation}) and $p_\mathrm{max}>0$ is the maximum transmit power. In this case,
	given a power vector $\signal{p}^\star\in\mathcal{P}$, we say that the network operates at an efficient operating point, in the weak Pareto sense, if there is no other power vector in $\mathcal{P}$ that strictly increases the utility of every user in the system; i.e., the set $\{\signal{p}\in\mathcal{P}\mid (\forall n\in\mathcal{N})~U_n(\signal{p})>U_n(\signal{p}^\star)\}$ is empty.  With utilities having the structure in \refeq{eq.utility}, efficient operating points can be easily parametrized in terms of the argument $\signal{p}\in\mathcal{P}$ of the utility functions. More precisely, under the constraint $\signal{p}\in\mathcal{P}$, the network with utilities $(U_1(\signal{p}),\ldots,U_N(\signal{p}))$ operates at an efficient point, in the weak Pareto sense, if and only if $\|\signal{p}\|=p_\mathrm{max}$ \cite{cavalcante2023}; i.e., the weak Pareto boundary of the utility region is given by $\{(U_1(\signal{p}),\ldots,U_N(\signal{p}))\in\real^N_+\mid \|\signal{p}\|=p_\mathrm{\max}\text{ and }\signal{p}\in\real^N_+\}$.
	
	Efficiency, in the sense defined above, is not enough from a system design perspective because the selected power vector also has to ensure fairness among users. In particular, operating points guaranteeing both efficiency and fairness can be obtained with the following optimization problem, which can be solved with simple fixed point algorithms such as those described in Fact~\ref{fact.cond_eig}(ii) in the Supplemental Material \cite{nuzman07,cavalcante2019}:
	\begin{align}
		\label{eq.maxmin}
		\begin{array}{rl}
			\mathrm{maximize}_{\signal{p}\in\mathcal{P}}&\min\limits_{n\in\mathcal{N}}\omega_n~U_n(\signal{p}),
		\end{array}
	\end{align} 
	where $(\omega_1,\ldots,\omega_N)\in\real_{++}^N$ are weights (priorities) assigned to users, and these weights establish the notion of fairness among users. For example, if $\omega_n=1$ for every $n\in\mathcal{N}$, we obtain the standard concept of max-min fairness, which has seen widespread application in cell-less networks in recent years \cite[Ch.~3, Ch.~7]{demir2021}\cite{miretti2025two}.\footnote{ Another common optimization criterion is sum-rate maximization, which can also be characterized as a solution to a problem of the form in \refeq{eq.maxmin} for appropriately chosen weights. However, finding weights leading to a power vector maximizing the sum rate is itself a challenging problem if the achievable rate region is nonconvex.}

	The solution to \refeq{eq.maxmin} ensures that no other fixed power allocation strategy strictly increases the utility of all users. However, with time-invariant utilities, the average utility of every user may be improved by allowing time sharing; i.e., scheduling different groups of users (or operating points) over time. Crucially, whether time sharing can increase the utility of every user relative to the fixed power allocation obtained from \refeq{eq.maxmin} -- regardless of the particular fairness weights adopted in \refeq{eq.maxmin} -- is governed by the geometry of the achievable utility region $\mathcal{S}_\mathcal{P}:=\{(U_1(\signal{p}),\ldots,U_N(\signal{p})) \in\real^N_+\mid \signal{p}\in\mathcal{P}\}$. To illustrate this link, consider $L \in \Natural$ given power vectors $\signal{p}_1, \ldots, \signal{p}_L$ from the set $\mathcal{P}$. Let $\alpha_l \geq 0$ denote the fraction of time allocated to the $l$th power vector, with $\sum_{l=1}^L \alpha_l = 1$. The resulting time-average utility $\bar{U}_n$ of user $n \in \mathcal{N}$ is then $\bar{U}_n = \sum_{l=1}^L \alpha_l\, U_n(\signal{p}_l).$ Therefore, the time-average utility vector $(\bar{U}_1,\ldots,\bar{U}_N)$ is a convex combination of points in $\mathcal{S}_\mathcal{P}$. If $\mathcal{S}_\mathcal{P}$ is convex, then  $(\bar{U}_1,\ldots,\bar{U}_N)\in\mathcal{S}_\mathcal{P}$, so there exists a single power vector $\signal{p}\in\mathcal{P}$ such that $(U_1(\signal{p}),\ldots, U_N(\signal{p}))=(\bar{U}_1,\ldots,\bar{U}_N)$; i.e., time sharing does not enlarge the achievable utility region. In contrast, if $\mathcal{S}_{\mathcal{P}}$ is nonconvex, the time-average utility vector $(\bar{U}_1,\ldots,\bar{U}_N)$ may lie outside $\mathcal{S}_\mathcal{P}$, meaning that this utility vector cannot be achieved by any single power vector $\signal{p}\in\mathcal{P}$.
	
	For concreteness, set the standard $l_\infty$-norm, denoted by \linebreak[4] $\|\cdot\|_\infty$, as the monotone norm in the set $\mathcal{P}$ of power constraints, and consider the hypothetical scenario with two users depicted in Fig.~\ref{fig.maxmin}, which shows a feasible utility (e.g., rate) region $\mathcal{S}_\mathcal{P}$. Point $C$ in the figure is the traditional max-min fair operating point by assuming that only adjustments to the power vector are possible. By allowing time sharing (link scheduling), we can alternate between the power vectors $\signal{p}_1=(p_\mathrm{max},~0)\in\mathcal{P}$ and $\signal{p}_2=(0,~p_\mathrm{max})\in\mathcal{P}$, using each vector half of the time. With this scheme, the users achieve, on average, the max-min utility corresponding to Point D in Fig.~\ref{fig.maxmin}, where the average utility of user $n\in\{1,2\}$ is $\frac{1}{2}U_n(\signal{p}_1)+\frac{1}{2}U_n(\signal{p}_2)$. As a result, the average utility of both users strictly increases, even though we only use power vectors from the set $\mathcal{P}$. In contrast, if the utility region were convex in this example, then, for the given weights defining the fairness criterion, no weighted max-min fair utility point obtainable via time sharing would strictly improve upon the utilities achieved by the solution to \refeq{eq.maxmin}. Therefore, convexity of the utility region serves as a certificate that time sharing offers no additional benefit.

\begin{figure}
	\begin{center}
		\includegraphics[width=0.6\columnwidth]{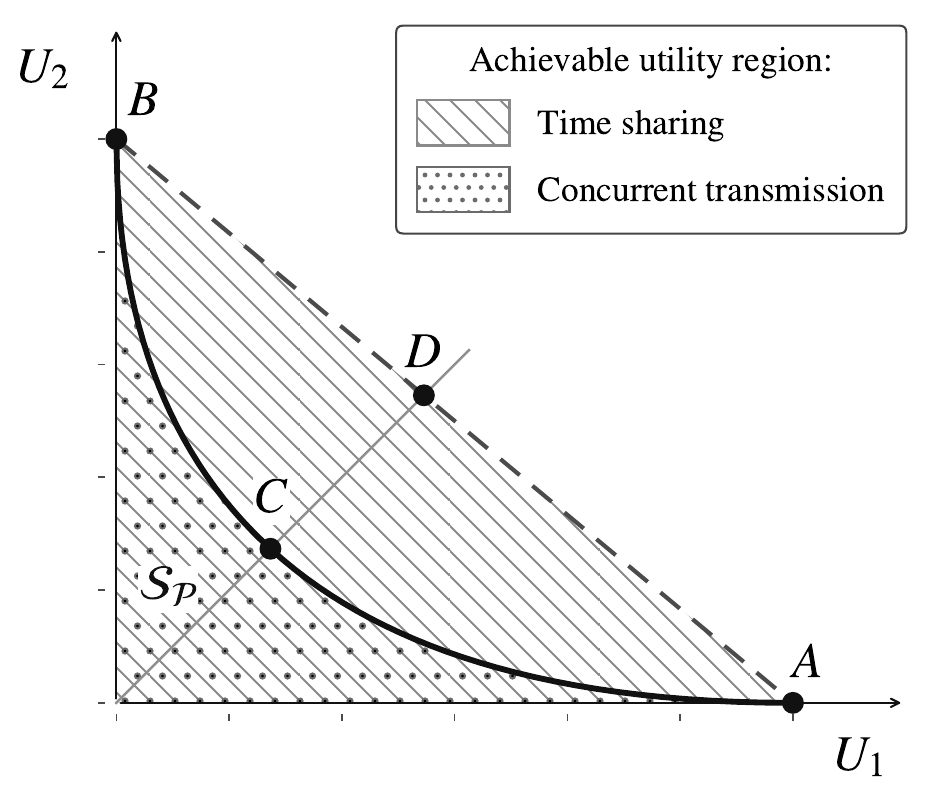}
		\caption{\footnotesize Utility for two users. Points A and B represent the utilities achieved with the power vectors $\signal{p}_1=(p_\mathrm{max},0)$ and $\signal{p}_2=(0,~p_\mathrm{max})$, respectively. Point C indicates the max-min fair utility point achieved with concurrent transmission, and Point D indicates the max-min fair utility point achieved through time sharing. Any utility below the line segment connecting points A and B (the convex hull of $\mathcal{S}_\mathcal{P}$) is achievable through time sharing.  
		}
		\vspace{-0.5cm}
		\label{fig.maxmin}
	\end{center}
	
\end{figure}

\subsection{Main contributions}

The first objective of this study is to link solutions to max-min problems to nonlinear eigenpairs of general interference mappings (see Sect.~\ref{sect.notation} for the definitions), thereby providing a concise description of the solutions (Remark~\ref{remark.maxminproblem}). The key step is a new result establishing a correspondence between conditional eigenpairs of standard interference mappings with eigenpairs of related general interference mappings (Proposition~\ref{proposition.cond_eig}). Building on this correspondence, we derive compact representations of the feasible SINR and rate regions of wireless networks, both with and without power constraints, in terms of the concept of nonlinear spectral radius -- see \refeq{eq.s_no_constraints}, \refeq{eq.R}, \refeq{eq.sp}, and \refeq{eq.rp}. A central implication of this characterization is that it allows us to derive sufficient conditions -- based on inverse Z-matrices -- that guarantee convexity of the SINR region, and, consequently, convexity of the rate region (Corollary~\ref{cor.setr} and Corollary~\ref{cor.withpc}). As discussed above, convexity has two important practical consequences. First, when these regions are convex, the performance attained by solutions to max-min problems cannot be improved via time sharing. Second, as shown in Sect.~\ref{subsect.sum}, sum-rate maximization can be formulated as a tractable convex optimization problem. Our framework covers a broad class of resource-allocation models in a unified manner, including those in \cite[Ch.~3, Ch.~7]{demir2021} for cell-less networks, \cite[Ch.~5.3.2]{marzetta16} and \cite[Theorem~7.1]{massivemimobook} for massive MIMO systems, and \cite[Ch.~5.6.3]{slawomir09} for traditional wireless networks, among others.

As further consequences of our results, we obtain a new characterization of the  weak Pareto boundary of the SINR and rate regions in terms of the spectral radius of nonlinear mappings (Proposition~\ref{proposition.wp}), resolve the conjecture in \cite[Sect.~5.4.4]{slawomir09}\cite{stanczak2007convexity} (Sect.~\ref{subsect.conjecture}), and introduce a new concept of channel compatibility among users that overcomes some limitations of the standard definition of favorable propagation in the massive MIMO literature (Sect.~\ref{sect.Zcompat}).

\subsection{Relation to existing studies}
{
To establish our contributions, we draw on fixed point theory in Thompson and Hilbert-projective metric spaces, linear algebra, and nonlinear analysis, using key results from  \cite{friedland81,nussbaum1986convexity,johnson2011,nuzman07,cavalcante2019,yates95}. Below, we list the main differences between our contributions and prior studies on wireless systems:

\begin{itemize}
\item Convexity of the SINR region for systems with affine interference functions has been studied in \cite[Ch.~5]{slawomir09}\cite{stanczak2007convexity}. In those studies, convexity is only established if the interference matrices described later in Sect.~\ref{sect.model} are  symmetric, which is an assumption unlikely to hold in practice. Our analysis lifts the symmetry assumption and incorporates power constraints. Furthermore, many of our results extend to nonlinear systems, which is essential for modeling networks employing optimal beamforming strategies.

\item Convexity of rate regions in massive MIMO systems without symmetry assumptions has been investigated in \cite{chen2018achievable}, but, unlike our results, the analysis is restricted to small systems with two users and simple power constraints. 

\item The results that follow are also related to those in \cite{tan2011maximizing,tan2011nonnegative,friedland2008maximizing,zheng2013maximizing,tsiaflakis2008optimality}, which are studies focusing on the sum-rate maximization problem using interference models that do not take into account self-interference or beamforming gain uncertainty, which has become a critical aspect in recent studies on extremely large MIMO arrays and distributed MIMO (and, in particular, cell-less) systems. Instead, our focus is on establishing sufficient conditions for the convexity of the entire SINR and rate regions of these modern interference models. By doing so, we identify a particular class of problems where the sum-rate maximization can be posed as a convex optimization problem. Notably, we do not impose symmetry on the interference patterns \cite[Sect.~V.A]{cheng2016optimal} or assume networks operating in the interference-free or in the low signal-to-noise ratio (SNR) regime \cite[Sect.~IV.B]{tan2011maximizing}. Our sufficient conditions for convexity of the SINR region (which are related to the concept of inverse Z-matrices, as also uncovered in \cite{tsiaflakis2008optimality} for interference models in digital subscriber lines) can identify convexity in scenarios where interference in some links is strong. Furthermore, unlike \cite{tan2011maximizing,tan2011nonnegative,friedland2008maximizing,zheng2013maximizing,tsiaflakis2008optimality}, we can cope with power constraints that are not necessarily convex polytopes, and we use the spectral radius of \emph{nonlinear} mappings to study interference of modern network models. This approach enables the presentation of the core ideas with compact and straightforward notation, significantly enhancing clarity and accessibility after the technical foundations have been established, even if we consider complex beamforming strategies. See, for example, Remark~\ref{remark.maxminproblem} in Sect.~\ref{sect.convexity} and  Remark~\ref{remark.simplicity} in Sect.~\ref{sect.withpower}.
\end{itemize}
}

\section{Mathematical framework}
\label{sect.preliminaries}

In this section, we establish notation and derive results that reveal the shape of the feasible SINR and rate regions for some wireless systems as corollaries. Owing to the mathematical nature of our study, Sect.~\ref{sect.notation} focuses on unifying terminology across various mathematical and wireless research domains, as consistent vocabulary is essential for the rigorous presentation of the main results. For completeness, the Supplemental Material lists existing results that are crucial for our proofs but not necessary for understanding the primary contributions of this study. We present the main results of Sect.~\ref{sect.preliminaries} in a dedicated subsection (Sect.~\ref{sect.convexity}). Readers interested in the implications of these findings to wireless systems may proceed directly to Sect.~\ref{sect.model} for a more application-oriented perspective.

\subsection{Notation and standard definitions}
\label{sect.notation}

We use the convention that the set $\Natural$ of natural numbers does not include zero. The sets of nonnegative and positive reals are denoted by, respectively, $\real_{+}:=~[0,\infty[$ and $\real_{++}:=~]0,\infty[$. Inequalities involving vectors should be understood coordinatewise. For example, given $(\signal{x}=(x_1,\ldots,x_N),\signal{y}=(y_1,\ldots,y_N))\in\real_+^N\times\real_+^N$, with $N\in\Natural$, we write $\signal{x}<\signal{y}$ if and only if $x_n < y_n$ for each $n\in\{1,\ldots,N\}$. A norm $\|\cdot\|$ in $\real^N$ is \emph{monotone} if $(\forall\signal{x}\in\real_+^N)(\forall\signal{y}\in\real_+^N)~\signal{x}\le\signal{y}\Rightarrow \|\signal{x}\|\le\|\signal{y}\|$. 	A norm $\|\cdot\|$ in $\real^N$ is \emph{polyhedral} if there exist $K\in\Natural$ vectors $\signal{a}_1,\ldots,\signal{a}_K$ in $\real^N_{+}\backslash\{\signal{0}\}$ such that $(\forall\signal{x}\in\real^N) \|\signal{x}\|=\max_{n\in\{1,\ldots,K\}}\signal{a}_n^t|\signal{x}|$, where $|\signal{x}|\in\real^N_+$ denotes the vector constructed with component-wise absolute values of $\signal{x}\in\real^N$ (e.g., the standard $l_1$ and $l_\infty$ norms are polyhedral). A sequence $(\signal{x}_n)_{n\in\Natural}$ of vectors in $\real^N$ is said to converge to $\signal{x}^\star$ if $\lim_{n\to\infty}\|\signal{x}_n-\signal{x}^\star\|=0$ for some norm $\|\cdot\|$ (and, hence, for every norm on $\real^N$ because of the equivalence of norms in finite dimensional spaces).

Let ${S}\subset \real_+^N$ be a nonempty set representing the utilities achievable by users in a network. In the context of the applications under consideration, a vector $\signal{x}\in{S}$ is said to lie on the weak Pareto boundary of ${S}$ if the set $\{\signal{y}\in\real_{++}^N\mid \signal{x}+\signal{y}\in{S}\}$ is empty; i.e., we cannot increase all components (utilities) of the vector $\signal{x}$ while staying within the set $\mathcal{S}$. A set ${S}\subset\real^N$ is convex if $(\forall\signal{x}\in\mathcal{S})(\forall\signal{y}\in{S})(\forall\alpha\in~]0,1[)~ \alpha\signal{x}+(1-\alpha)\signal{y}\in{S},$ and it is called downward comprehensive in $\real_{+}^N$ if $(\forall \signal{x}\in S)(\forall \signal{y}\in \real_+^N) \signal{y}\le\signal{x}\Rightarrow \signal{y}\in S$. The closure of a set $S\subset \real^N$ is denoted by $\overline{S}$.

A diagonal matrix with positive components on the diagonal is a \emph{positive diagonal matrix}, and the set of nonnegative and positive diagonal matrices of dimension $N\times N$ are denoted by, respectively, $\mathcal{D}_+^{N\times N}$ and $\mathcal{D}_{++}^{N\times N}$. By \linebreak[4] $\mathrm{diag}(x_1,\ldots,x_N)=\mathrm{diag}(\signal{x})$, where $\signal{x}=(x_1,\ldots,x_N)\in\real^N$, we denote an $\real^{N\times N}$ diagonal matrix having the scalars $x_1,\ldots,x_N$ on the diagonal entries. We write the transpose of a matrix $\signal{M}$ as $\signal{M}^t$, and the matrix exponential as $e^\signal{M}$. Any matrix $\signal{M}\in\real^{N\times N}$ can be interpreted as a linear mapping $\real^N\to\real^N: \signal{x}\mapsto\signal{Mx}$ (standard basis), so we write $\signal{M}T$ to denote the mapping $\mathcal{X}\to\real^N:\signal{x}\mapsto (\signal{M} \circ T)(\signal{x})=\signal{M}(T(\signal{x}))$, where $T:\mathcal{X}\to\mathcal{X}$ and $\mathcal{X}\subset\real^N$. A matrix $\signal{M}\in\real^{N\times N}$ is called a \emph{Z-matrix} if it has nonpositive off-diagonal components \cite{nussbaum1986convexity}\cite{johnson2011}. If a matrix $\signal{M}\in\real^{N\times N}$ has an inverse $\signal{M}^{-1}$ that is a Z-matrix, then $\signal{M}$ is said to be an \emph{inverse Z-matrix}. A matrix $\signal{M}\in\real^{N\times N}$ is called an M-matrix if it can be expressed as $\signal{M}=\alpha\signal{I}-\signal{P}$ for some $\signal{P}\in\real^{N\times N}_+$ and $\alpha>\rho$, where $\rho$ is the spectral radius of $\signal{P}$, in the conventional sense in linear algebra. If $\signal{M}^{-1}$ is an M-matrix, then $\signal{M}$ is called an \emph{inverse M-matrix.} 

For convenience, given $(N, M)\in\Natural\times\Natural$ and two sets $\mathcal{X}\subset\real^N$ and $\mathcal{Y}\subset\real^M$, we say that the mapping $f:\mathcal{X}\to\mathcal{Y}$ is a function only if $M=1$ (NOTE: if $M=1$, $f$ is both a function and a mapping). A mapping $f:\mathcal{X}\subset\real^N\to\mathcal{Y}\subset\real^M$ is \emph{continuous} if, for every sequence $(\signal{x}_n)_{n\in\Natural}$ in $\mathcal{X}$ converging to $\signal{x}^\star\in\mathcal{X}$, the sequence $(f(\signal{x}_n))_{n\in\Natural}$ in $\mathcal{Y}$ converges to $f(\signal{x}^\star)\in\mathcal{Y}$. Let ${S}\subset\real^N$ be a nonempty convex set. A function  $f:{S}\to\real$ is \emph{convex} if $(\forall\alpha\in~]0,1[)(\forall\signal{x}\in{S})(\forall\signal{y}\in{S})~f(\alpha \signal{x}+(1-\alpha)\signal{y})\le \alpha f(\signal{x})+(1-\alpha)f(\signal{y}),$ and recall that its level set $\{\signal{x}\in{S}~|~f(\signal{x})\le \beta\}$  is a (possibly empty) convex set for every $\beta\in\real$. A function $f:\real_+^N\to\real_{+}$ is \emph{monotonic} or \emph{monotone} if $(\forall\signal{x}\in\real_+^N)(\forall\signal{y}\in\real_+^N)~\signal{x}\le\signal{y}\Rightarrow f(\signal{x})\le f(\signal{y})$; \emph{scalable} if $(\forall\signal{x}\in\real_+^N)(\forall\alpha>1)~ f(\alpha\signal{x}) < \alpha f(\signal{x})$; and \emph{positively homogeneous} if $(\forall\signal{x}\in\real_+^N)(\forall\alpha>0)~ f(\alpha\signal{x}) = \alpha f(\signal{x})$.  Following standard terminology in the wireless literature, we introduce below two classes of mappings that have played a major role in network optimization. 

\begin{definition}
	\label{def.sif}
	(Standard and general interference functions) A \emph{continuous} function $t:\real_+^N\to\real_{++}$ is called a \emph{standard interference function} if it is both monotone and scalable \cite{yates95}; whereas a \emph{continuous} function $t:\real_+^N\to\real_{+}$ is called a \emph{general interference function} if it is both monotone and positively homogeneous \cite{boche2008,martin11}. Likewise, a mapping $T:\real_+^N\to\real_{+}^N:\signal{x}\mapsto(t_1(\signal{x}),\ldots,t_N(\signal{x}))$ is called a standard interference mapping (respectively, general interference mapping) if the coordinate functions $(t_n)_{n\in\{1,\ldots,N\}}$ are standard interference functions (respectively, general interference functions). 
\end{definition}

We recall from \cite{burbanks2003extension} that monotonicity and either scalability or homogeneity on the cone $\real_{+}^N$ implies continuity on $\real_{++}^N$, but not necessarily on the boundary of the domain $\real_{+}^N$. For this reason, we assume continuity everywhere in the above definition to avoid technical digressions. Furthermore, we have lifted some of original restrictions related to the domain and codomain of general interference functions. For example, unlike the original definition in \cite{boche2008,martin11}, the function $f:\real_+\to\real_+:x\mapsto 0$ is a general interference function according to our definition. It is well known that monotonicity and scalability imply positivity of a mapping in the entire nonnegative cone $\real_+^N$, so we often use $\real_{++}^N$ to denote the codomain of a standard interference mapping to emphasize this fact. From an application perspective, an important proper subclass of standard interference mappings is the class of \emph{continuous positive concave mappings} (see~\cite[Proposition~1]{cavalcante2016}); i.e., continuous mappings $T:\real_{+}^N\to\real_{++}^N$ such that $(\forall\signal{x}\in\real_+^N)(\forall\signal{y}\in\real_+^N)(\forall\alpha\in~]0,1[)~ T(\alpha\signal{x}+(1-\alpha)\signal{y})\ge \alpha T(\signal{x})+(1-\alpha)T(\signal{y}).$

The set of fixed points of a mapping $T:\real_+^N\to\real_+^N$ is denoted by $\mathrm{Fix}(T):=\{\signal{x}\in\real^N_+\mid T(\signal{x})=\signal{x}\}$, and the mapping defined below plays a crucial role in the analysis of the fixed point set of standard interference mappings (see, for example, Fact~\ref{fact.tinf} in the Supplemental Material):

\begin{definition}
	\label{def.am}
	(\cite{cavalcante2019,oshime92} Asymptotic mappings) We associate to a standard interference mapping $T:\real^N_+\to\real_{++}^N$ a general interference mapping $T_\infty:\real_+^N\to\real_+^N$, called asymptotic mapping, defined for every $\signal{x}\in\real^N_+$ by \linebreak[4] $T_\infty(\signal{x}):=\lim_{h\to\infty}(1/h)T(h\signal{x})$, and we recall that this limit always exists.
\end{definition}

 If $(\signal{x},\lambda)\in\real_{+}^N\backslash\{\signal{0}\}\times\real_+$ satisfies $T(\signal{x})=\lambda\signal{x}$ for a mapping $T:\real^N_+\to\real^N_+$, we say that $\signal{x}$ is an (nonlinear) eigenvector associated with the (nonlinear) eigenvalue $\lambda$ of $T$. If $T:\real_+^N\to\real_{++}^N$ is a standard interference mapping and $\|\cdot\|$ is a monotone norm on $\real^N$, then there exists a unique eigenpair $(\signal{x}, \lambda)\in\real_{++}^N\times\real_{++}$ satisfying $T(\signal{x})=\lambda \signal{x}$ and $\|\signal{x}\|=1$ \cite{nuzman07} (Fact 6(i) in the Supplemental Material). In this case, we call $\signal{x}$ and $\lambda$, respectively, the (nonlinear) conditional eigenvector and the (nonlinear) conditional eigenvalue associated with $T$ and $\|\cdot\|$.  If the monotone norm is clear from the context, we simply call $\signal{x}$ and $\lambda$ the conditional eigenvector and eigenvalue of $T$. The spectral radius of general interference mappings is defined as the supremum of all eigenvalues: 

\begin{definition}
	\label{def.nl_radius}(\cite[Definition~3.2]{nussbaum1986convexity} Spectral radius) Let \linebreak[4] $G:\real^N_+\to\real^N_+$ be a general interference mapping. The spectral radius of $G$ is defined to be 
	\begin{align}
		\label{eq.nl_radius}
		\rho(G) := \sup\{\lambda\in\real_+~|~(\exists \signal{x}\in\real_+^N\backslash\{\signal{0}\})~ G(\signal{x})=\lambda\signal{x}\}.
	\end{align}
	(NOTE: There always exists an eigenvector $\signal{x}^\star$ in $\real_+^{N}\backslash\{\signal{0}\}$ such that $G(\signal{x}^\star)=\rho(G)\signal{x}^\star$; i.e., the supremum in \refeq{eq.nl_radius} is attained  \label{fact.achieved} \cite[Proposition~5.3.2(ii) and Corollary~5.4.2]{lem13}.)
\end{definition}

We recall that there exist simple numerical schemes able to compute the spectral radius \cite[Remark~3]{cavalcante2019}\cite{krause01}.  In particular,  if the general interference mapping is linear, then it can be written as $\signal{x}\mapsto \signal{Mx}$ for some nonnegative matrix $\signal{M}\in\real_{+}^{N\times N}$, and we can use classical results from Perron-Frobenius theory to confirm that $\rho(\signal{M})$ in Definition~\ref{def.nl_radius}, with $\signal{M}$ interpreted as a linear mapping, coincides with the usual notion of the spectral radius of the matrix $\signal{M}$ in linear algebra.

Having established the key terminology and concepts, we can proceed with our main contributions.

\subsection{Relating the conditional eigenvalue of standard interference mappings to the spectral radius of general interference mappings}
\label{sect.convexity}

To motivate the results in this section, consider the following optimization problem:

\begin{example}\label{example.maxmin} Let $K\in\Natural$, $N\in\Natural$, $x_\mathrm{max}>0$, $[\signal{a}_1,\ldots,\signal{a}_K]=\signal{A}\in\real_+^{N\times K}$, $[b_1,\ldots,b_N]^t=\signal{b}\in\real_{++}^N$, $[\signal{c}_1,\ldots,\signal{c}_N]=\signal{C}\in\real_+^{N\times N}$, and $[\sigma_1,\ldots,\sigma_N]^t=\signal{\sigma}\in\real_{++}^N$ be given problem parameters. For notational simplicity, define $\signal{M}:=\mathrm{diag}(\signal{b})^{-1}\signal{C}^t$, $\signal{u}:=\mathrm{diag}(\signal{b})^{-1}\signal{\sigma}$, and $(\forall n\in\{1,\ldots,K\})~\signal{M}_n:=\signal{M}+(1/x_\mathrm{max})\signal{u}\signal{a}_n^t$.	Assume that the function $\real^N\to\real_+:\signal{x}\mapsto \dfrac{1}{x_\mathrm{max}}\max_{n\in\{1,\ldots,K\}}\signal{a}_n^t|\signal{x}|=:\|\signal{x}\|$ is a polyhedral monotone norm, and consider the optimization problem:
\begin{align}
	\label{eq.maxmin_simple}
	\begin{array}{rl}
		\text{maximize}_{\signal{x}=(x_1,\ldots,x_N)\in\real^N} & \min_{n\in\{1,\ldots,N\}} \dfrac{b_n x_n}{\signal{c}_n^t\signal{x}+\sigma_n}\\
		\text{s.t.} & (\forall n\in\{1,\ldots,K\})~ \signal{a}_n^t\signal{x}\le x_\mathrm{max} \\
		& \signal{x}\ge\signal{0}.	
	\end{array}
\end{align}
 Standard results in the literature \cite[Ch.~5.6.3]{slawomir09}\cite[Proposition~4]{miretti2022closed} show that a (not necessarily unique) solution $\signal{x}^\star\in\real_{++}^N$ to Problem~\refeq{eq.maxmin_simple} is the eigenvector (normalized to $\|\signal{x}^\star\| = 1$) associated with the largest eigenvalue of any matrix $\signal{M}_{n^\star}$ with index $n^\star\in\{1,\ldots,K\}$ satisfying $\rho(\signal{M}_{n^\star})=1/t^\star$, where $t^\star:=1/(\max_{n\in\{1,\ldots,K\}}\rho(\signal{M}_n))$ is the optimal objective of \refeq{eq.maxmin_simple}.
\end{example}

{
	Readers familiar with wireless systems may recognize Problem~\refeq{eq.maxmin_simple} as a max-min SINR optimization problem. Formulations in  \cite{miretti2022closed}\cite[Ch.~5.6.3]{slawomir09}\cite[Ch.~5.3.2]{marzetta16}\cite[Theorem~7.1]{massivemimobook} arise as special cases. To illustrate, consider the uplink of a system with $N\in\Natural$ users and a single access point. Then the optimization variable $\signal{x}$ represents the transmit power vector; $(b_n)_{n\in\{1,\ldots,N\}}$ are the effective channel gains; $(\signal{c}_n)_{n\in\{1,\ldots,N\}}$ are vectors collecting the effective interference channels; $(\sigma_n)_{n\in\{1,\ldots,N\}}$ model the receiver noise levels; and $(\signal{a}_n)_{n\in\{1,\ldots,K\}}$ encode power constraints. For instance, with $K=N$, choosing $(\signal{a}_n)_{n\in\{1,\ldots,N\}}$ as the standard basis vectors yields the per-user power constraint $\|\signal{x}\|_\infty \le x_{\max}$, where $x_\mathrm{max}$ is the maximum transmit power. We postpone a detailed discussion in more general settings to the next sections.} 

From a high-level perspective, the characterization of the solution to Problem~\refeq{eq.maxmin_simple} provided in Example~\ref{example.maxmin} is somewhat cumbersome: it requires comparing the spectral radii of an entire family of matrices. As a first concrete application of the main results of this section, we show that the same solution can instead be expressed more naturally through the eigenpair of a possibly nonlinear general interference mapping. In particular, we later establish the following alternative characterization:

\begin{remark}
	\label{remark.maxminproblem}
	Define a (possibly nonlinear) general interference mapping by $T_{\|\cdot\|}:\real_+^N\to\real_{+}^N:\signal{x}\mapsto \signal{Mx}+\signal{u}\|\signal{x}\|$, where $\signal{M}$, $\signal{u}$, $x_\mathrm{max}$, and $\|\cdot\|$ are as in Example~\ref{example.maxmin}. A solution to Problem~\refeq{eq.maxmin_simple} is the eigenvector $\signal{x}^\star\in\real_{++}^N$ of $\tnorm$, normalized to satisfy $\|\signal{x}^\star\|=1$, associated with the nonlinear spectral radius $\rho(T_{\|\cdot\|})>0$, and $1/\rho(T_{\|\cdot\|})$ is the optimal objective of Problem~\refeq{eq.maxmin_simple}. 
\end{remark}

The above nonlinear viewpoint has several advantages. It not only allows us a compact description of solutions to many problems in the wireless literature, but it also enables us to extend the matrix-based results of Example \ref{example.maxmin} to genuinely nonlinear scenarios, such as those encountered in wireless systems employing optimal beamforming strategies \cite{miretti2025two}. Furthermore, it enables the use of existing numerical methods for computing nonlinear eigenvectors to solve variants of Problem~\refeq{eq.maxmin_simple}, which may prove valuable in future research. Most importantly for the developments that follow, the framework we develop provides the mathematical machinery necessary to analyze the geometry of utility regions in wireless networks.

To derive the results mentioned in Remark~\ref{remark.maxminproblem}, we begin by recalling a standard fact from the wireless literature: the pointwise minimum of a finite family of standard interference functions results in a standard interference function \cite{yates95}. However, this closure property does not generally extend to the pointwise infimum of an infinite family. A simple counterexample is given by the constant functions $f_{c}:\real_+^N\to\real_{++}:\signal{x}\mapsto c$ parametrized by $c>0$. Each function $f_c$ with $c>0$ is a standard interference function, but the pointwise infimum $(\forall\signal{x}\in\real_+^N)~f(\signal{x}):=\inf_{c>0}f_c(\signal{x})=0$ fails to satisfy the scalability property. Nevertheless, in wireless systems, the families of functions typically possess additional structure -- such as that described in the next lemma -- ensuring that the pointwise infimum is again a standard interference function.

\begin{lemma}
	\label{lemma.f}
	Let $\mathcal{Y}$ be a nonempty set. Suppose that, for each $y\in\mathcal{Y}$, $h_{y}:\real_+^N\to\real_+$ is a general interference function and ${u}_y\in\real_{++}$ is a positive scalar. Assume that there exists $\delta>0$ satisfying $(\forall y\in\mathcal{Y})~ {u}_{y}\ge {\delta}>{0}$; i.e., the set \linebreak[4] $\{ u_y\in\real_{++}~|~y\in\mathcal{Y}\}\subset\real_{++}$  is bounded away from zero. Define a function $t:\real_+^N\to\real_{++}$ by
	\begin{align}
		\label{eq.f}
		(\forall\signal{x}\in\real_+^N)~t(\signal{x}):=\inf_{y\in\mathcal{Y}}(h_{y}(\signal{x})+{u}_{y}).
	\end{align} 
	In addition, given a monotone norm $\|\cdot\|$  in $\real^N$, define the  function  $t_{\|\cdot\|}:\real_+^N\to\real_{+}$ associated with $t$ by
	\begin{align}
		\label{eq.g}
		(\forall\signal{x}\in\real^N_+)~t_{\|\cdot\|}(\signal{x}):=\inf_{y\in\mathcal{Y}}(h_{y}(\signal{x})+{u}_{y}~\|\signal{x}\|).
	\end{align}
Assume that $t$ and $t_{\|\cdot\|}$ are continuous on $\real_+^N$. Then $t$ in \refeq{eq.f} is a standard interference function, and $t_{\|\cdot\|}$ in \refeq{eq.g} is a general interference function.
\end{lemma}

\begin{proof}
	Let $y\in\mathcal{Y}$ be arbitrary and $(\signal{x},\signal{z})\in\real_+^N\times\real_+^N$ satisfy  $\signal{x}\le\signal{z}$, so we have $ h_y(\signal{x})+u_y \le h_y(\signal{z})+u_y$ because general interference functions are monotone. This last inequality is valid for every $y\in\mathcal{Y}$, so we have $t(\signal{x})=\inf_{y\in\mathcal{Y}} (h_y(\signal{x})+u_y) \le \inf_{y\in\mathcal{Y}} (h_y(\signal{z})+u_y)=t(\signal{z})$, which proves monotonicity. Now, take $y\in\mathcal{Y}$, $\signal{x}\in\real^N_+$, and $\alpha>1$. Homogeneity of $h_y$ yields $h_y(\alpha\signal{x})+u_y = \alpha h_y(\signal{x})+u_y$.  The assumption $(\forall y\in\mathcal{Y})~ u_y\ge \delta > 0$ and $\alpha>1$ implies that $0\le h_y(\alpha\signal{x})+u_y-\delta \le \alpha (h_y(\signal{x}) + u_y-\delta) $. Therefore, $h_y(\alpha\signal{x})+u_y+ (\alpha-1) \delta \le \alpha (h_y(\signal{x}) + u_y)$, and we conclude that $t(\alpha\signal{x})+(\alpha-1)\delta=\inf_{y\in \mathcal{Y}}(h_y(\alpha\signal{x})+u_y) +(\alpha-1)\delta \le \inf_{y\in \mathcal{Y}}\alpha (h_y(\signal{x}) + u_y) = \alpha t(\signal{x})$, which implies ${t(\alpha\signal{x})}<\alpha t(\signal{x})$ because $\delta>0$ and $\alpha>1$. 
	The proof that $t_{\|\cdot´\|}$ is a general interference mapping is simpler because it does not rely on positivity of $u_y$. Since it follows with similar steps, we omit the proof for brevity.
\end{proof}

We use the functions in \refeq{eq.f} and \refeq{eq.g} in Lemma~\ref{lemma.f} to define the following proper subclass of standard and general interference mappings:

\begin{definition} 
	\label{def.compatible}
	Let $\mathcal{N}:=\{1,\ldots,N\}$. We say that a standard interference mapping $T:\real_+^N\to\real_{++}^N:\signal{x}\mapsto (t_1(\signal{x}),\ldots,t_N(\signal{x}))$ and a general interference mapping $T_{\|\cdot\|}:\real^N_+\to\real_{+}^N:\signal{x}\mapsto (\tn{1}(\signal{x}),\ldots,\tn{N}(\signal{x}))$ satisfy hypothesis ($T$,$\tnorm$) for a monotone norm  $\|\cdot\|$ on $\real^N$ if there exist sets $(\mathcal{Y}_n)_{n\in\mathcal{N}}$, a family $(u_{y,n})_{n\in\mathcal{N},y\in\mathcal{Y}_n}$ of positive scalars bounded away from zero, and general interference functions $(h_{y,n})_{n\in\mathcal{N},y\in\mathcal{Y}_n}$ such that the coordinate functions $t_n$ and $(t_n)_{\|\cdot\|}$ for each $n\in\mathcal{N}$ can be expressed as
	\begin{align} 
		\label{eq.siinf}
		t_n:\real_+^N\to\real_{++}:\signal{x}\mapsto\inf_{y\in\mathcal{Y}_n}(h_{y,n}(\signal{x})+{u}_{y,n})
	\end{align}
	and  
\begin{multline}
	\label{eq.tnorm}
	(t_n)_{\|\cdot\|}:\real^N_+\to\real_+:\signal{x}\mapsto\inf_{y\in\mathcal{Y}_n}(h_{y,n}(\signal{x})+{u}_{y,n}~\|\signal{x}\|).
\end{multline}
\end{definition}

As a simple illustration of mappings satisfying hypothesis ($T$,$\tnorm$) for a given monotone norm $\|\cdot\|$, we have: 

\begin{example} \label{example.illustration} Under the stated assumptions on $\signal{M}$, $\signal{u}=:[u_1,\ldots,u_N]^t$, $\mathcal{N}:=\{1,\ldots,N\}$, and the monotone norm $\|\cdot\|$ in Example~\ref{example.maxmin}, we verify that $T:\real_+^N\to\real_{++}^N:\signal{x}\mapsto \signal{Mx}+\signal{u}$ and $T_{\|\cdot\|}:\real_+^N\to\real_{+}^N:\signal{x}\mapsto \signal{Mx}+\signal{u}\|\signal{x}\|$ satisfy hypothesis $(T,\tnorm)$. More specifically, the $n$th coordinate functions $(t_n)_{n\in\mathcal{N}}$ and $(t_n)_{\|\cdot\|}$ of, respectively, $T$ and $T_{\|\cdot\|}$ are given by $t_n:\real^N_{+}\to\real_{++}: \signal{x}\mapsto h_{n}(\signal{x})+{u}_{n}$ and $(t_n)_{\|\cdot\|}:\real^N_{+}\to\real_{+}: \signal{x}\mapsto h_{n}(\signal{x})+{u}_{n}\|\signal{x}\|$, where $(\forall n\in\mathcal{N})~h_n:\real^N_+\to\real_+:\signal{x}\mapsto \signal{m}_n^t\signal{x}$ is a general interference function and $\signal{m}_n^t$ is a row vector corresponding to the $n$th row of the matrix $\signal{M}$. 
\end{example}

Let $T$ and $\tnorm$ satisfy hypothesis $(T,\tnorm)$ for a given norm $\|\cdot\|$. We now proceed to show that important properties of conditional eigenvalues -- which are typically related to SINR values or rates in wireless network models \cite{nuzman07,cavalcante2019} -- associated with the monotone norm and the mapping $T$ can be deduced from the spectral radius $\rho(T_{\|\cdot\|})$ and the corresponding eigenvector of $T_{\|\cdot\|}$. The mapping $T_{\|\cdot\|}$ unifies $T$ and the norm $\|\cdot\|$ into a single concept that is often more convenient for analysis. The next proposition also shows that $T_{\|\cdot\|}$ possesses a unique eigenvector (up to a scaling factor) associated with the spectral radius, which is a special result considering that we deal with possibly nonlinear mappings.

\begin{proposition}
	\label{proposition.cond_eig}
	Given a monotone norm $\|\cdot\|$ on $\real^N$, assume that $T:\real_+^N\to\real_{++}^N$ and $\tnorm:\real_+^N\to\real_{+}^N$ satisfy hypothesis ($T$,$\tnorm$). Let $\signal{x}^\star\in\real_{++}^N$ and $\lambda^\star>0$ be, respectively, the conditional eigenvector and the conditional eigenvalue associated with $T$ and $\|\cdot\|$ (see Fact~\ref{fact.cond_eig}(i) in the Supplemental Material). Then each of the following holds:
	\item[(i)] $\rho(\tnorm)=\lambda^\star$, where we recall that $\rho(\tnorm)$ is the spectral radius of the mapping $\tnorm$ in Definition~\ref{def.compatible}; 
	\item[(ii)] if $\signal{z}\in\real_{+}^N\backslash\{\signal{0}\}$ is an eigenvector associated with the spectral radius $\rho(T_{\|\cdot\|})$ of $\tnorm$ (i.e., $\tnorm(\signal{z})=\rho(T_{\|\cdot\|})\signal{z}$), then $\signal{x}^\star=(1/\|\signal{z}\|)\signal{z}\in\real_{++}^N$; and
	\item[(iii)] the set of eigenvectors of $\tnorm$ associated with the spectral radius $\rho(\tnorm)$ is the cone given by $\{\alpha\signal{x}^\star\in\real_{++}^N \mid \alpha>0\}$.	
\end{proposition}
\begin{proof} (i)-(ii)
	Assume that $\signal{z}\in\real^N_{+}\backslash\{\signal{0}\}$ is an eigenvector associated with the spectral radius $\rho(\tnorm)\ge 0$ of $T_{\|\cdot\|}$. Let $\signal{w}$ be the normalized version of $\signal{z}$ defined by $\signal{w}:=({1}/{\|\signal{z}\|})\signal{z}$. As a result,
	\begin{align}
		\label{eq.parti}
		\rho(\tnorm) \signal{w} \overset{(a)}= T_{\|\cdot\|}(\signal{w}) \overset{(b)}=T(\signal{w})\overset{(c)}>\signal{0},
	\end{align}  
	where (a) follows from the fact that $T_{\|\cdot\|}$ is a general interference mapping (Lemma~\ref{lemma.f}), and, hence, positively homogeneous; (b) follows from $\|\signal{w}\|=1$ and the definitions of $T$ and $\tnorm$; and (c) follows from positivity of $T$. The above shows that $\signal{w}$ and $\rho(\tnorm)$ are, respectively, a conditional eigenvector and a conditional eigenvalue associated with $T$ and the monotone norm $\|\cdot\|$. Uniqueness of the conditional eigenvalue and eigenvector (Fact~\ref{fact.cond_eig}(i) in the Supplemental Material) implies that $\rho(\tnorm)=\lambda^\star$ and $\signal{w}=(1/\|\signal{z}\|)\signal{z}=\signal{x}^\star$, which completes the proof of parts (i) and (ii).
	
	(iii) Multiply both sides of equality (a) in \refeq{eq.parti} by $\alpha>0$ and use positive homogeneity of $\tnorm$ and $\signal{w}=\signal{x}^\star$ (as proved in part (ii)) to conclude the proof.
\end{proof}

We are now in a position to make precise the claims stated in Remark~\ref{remark.maxminproblem}.

\begin{example}
	\label{example.clarification}
	Consider the mappings $T$ and $\tnorm$ defined in Example~\ref{example.illustration} under the assumptions stated on $\signal{M}$, $\signal{u}$, $x_\mathrm{max}$, and $\|\cdot\|$ in Example~\ref{example.maxmin}. We verify from \cite[Eq. (10)]{miretti2022closed} that a solution $\signal{x}^\star\in\real_{++}^N$ to Problem~\refeq{eq.maxmin_simple} satisfies $T(\signal{x}^\star)=(1/t^\star) \signal{x}^\star$ and $\|\signal{x}^\star\|=1$, where $t^\star>0$ is the optimal objective of Problem~\refeq{eq.maxmin_simple}.  As a result, $\signal{x}^\star$ and $1/t^\star$ are, respectively, the conditional eigenvector and eigenvalue corresponding to the mapping $T$ and the monotone norm $\|\cdot\|$ in Example~\ref{example.maxmin}. The claims in Remark~\ref{remark.maxminproblem} now follow immediately from Proposition~\ref{proposition.cond_eig}.
\end{example}

In the remainder of this section, for a given monotone norm $\|\cdot\|$ and mappings $T$ and $\tnorm$ satisfying hypothesis ($T$,$\tnorm$), we establish conditions under which the function $g_{\tnorm}: \real_+^N\to\real_+:\signal{x}\mapsto\rho(\diagvec{x}T_{\|\cdot\|})$ -- \emph{i.e.,  the function mapping a given $\signal{x}\in\real_+^N$ to the spectral radius of the general interference mapping $\real^N_+\to\real^N_+:\signal{y}\mapsto \mathrm{diag}(\signal{x})(\tnorm(\signal{y}))=(\mathrm{diag}(\signal{x})\circ \tnorm)(\signal{y})$} --\footnote{Here, given $\signal{x}\in\real_+^N$, the matrix $\mathrm{diag}(\signal{x})\in\real_+^{N\times N}$ should be interpreted as the linear operator $\mathrm{diag}(\signal{x}):\real_{+}^{N}\to\real_+^N:\signal{y}\mapsto \mathrm{diag}(\signal{x})\signal{y}$.} satisfies all properties of a norm on the nonnegative cone $\real^N_+$. We then extend this cone-defined norm to the entire space \(\mathbb{R}^N\). In the next section, we show that this extension enables us to analyze interference models in wireless networks using concise notation that immediately reveals high-level properties of the feasible SINR and rate regions. For instance, by simply glancing at the definition of the feasible SINR region in \refeq{eq.csp} in Sect.~\ref{sect.withpower} below, readers can infer from the notation that this region resembles a ball defined by some norm, making it both convex and compact, among other intuitive properties. Remarkably, this insight is achieved without requiring any knowledge of the specific details of the interference model, the wireless technology, or the norm under consideration.

\begin{proposition}
	\label{proposition.qnorm}
	Let $\|\cdot\|$ be a monotone norm and \linebreak[4] $T:\real^N_+\to\real^N_{++}$ and $\tnorm:\real^N_+\to\real^N_{+}$ satisfy hypothesis ($T$,$\tnorm$). Consider the function $\ngn:\real^N_+\to\real_+:\signal{x}\mapsto\rho(\mathrm{diag}(\signal{x})T_{\|\cdot\|})$. Then each of the following holds:
	\begin{itemize}
		\item[(i)] $(\forall \alpha>0)(\forall \signal{x}\in\real^N_+)~ \ngn(\alpha\signal{x})=\alpha \ngn(\signal{x})$;
		\item[(ii)] $(\forall \signal{x}\in\real^N_+)~\ngn(\signal{x})=0\Leftrightarrow \signal{x}=\signal{0}$;
		\item[(iii)] $(\forall \signal{x}\in\real_+^N)(\forall \signal{y}\in\real_+^N)~\signal{x}\le\signal{y}\Rightarrow \ngn(\signal{x})\le \ngn(\signal{y})$; and
		\item[(iv)] $\ngn$ is continuous on $\real^N_+$. 
	\end{itemize}
\end{proposition}
\begin{proof}
	(i) Immediate from the positive homogeneity of the spectral radius.
	
	(ii) The direction $\signal{x}=\signal{0} \Rightarrow \ngn(\signal{0})=0$ is immediate from the definition of the spectral radius (Definition~\ref{def.nl_radius}), so we only need to prove the converse $ \ngn(\signal{x})=0 \Rightarrow \signal{x}=\signal{0}$, which is equivalent to the contrapositive $(\forall \signal{x}\in\real_+^N)~\signal{x}\ne\signal{0}\Rightarrow \ngn(\signal{x})>0$. Choose $\signal{x}=(x_1,\ldots,x_N)\in\real^N_{+}\backslash\{\signal{0}\}$ arbitrarily. Let $\widetilde{T}:\real^M_+\to\real_{++}^M$ be the mapping obtained by removing all coordinates $n\in\mathcal\{1,\ldots,N\}$ of $T$ for which $x_n=0$, so we have $1\le M \le N$. Similarly, let $\widetilde{\signal{x}}\in\real^M_{++}$ be the reduced vector obtained by removing all zero components of $\signal{x}$. For every vector $\tilde{\signal{y}}\in\real^M$ constructed with the nonzero entries of a vector $\signal{y}\in\real^N$  with the same zero pattern of $\signal{x}$, define $\|\widetilde{\signal{y}}\|_\star:=\|\signal{y}\|$, and note that $\|\cdot\|_\star$ is a monotone norm on $\real^M_+$. Since the coordinate functions of $\widetilde{T}$ can be written as in \refeq{eq.siinf} (because $T$ and $\tnorm$ satisfy hypothesis $(T,\tnorm)$), so does the mapping $\signal{D}\tilde{T}$, where $\signal{D}:=\mathrm{diag}(\widetilde{\signal{x}})\in\mathcal{D}_{++}^{M\times M}$. For the mapping $\signal{D}\tilde{T}$ written in the form in \refeq{eq.siinf}, we denote by $(\signal{D}\widetilde{T})_{\|\cdot\|_\star}=\signal{D}\widetilde{T}_{\|\cdot\|_\star}$ the corresponding mapping with the coordinate functions written as in \refeq{eq.tnorm}. We deduce from Proposition~\ref{proposition.cond_eig} that the spectral radius of $(\signal{D}\widetilde{T})_{\|\cdot\|_\star}$ satisfies $\rho((\signal{D}\widetilde{T})_{\|\cdot\|_\star})>0$, and there exists a positive vector $\widetilde{\signal{u}}\in\real_{++}^M$ satisfying $(\signal{D}\widetilde{T})_{\|\cdot\|_\star}(\widetilde{\signal{u}})=\rho((\signal{D}\widetilde{T})_{\|\cdot\|_\star})\widetilde{\signal{u}}$. Denote by $\signal{u}\in\real^N_+$ the vector obtained from $\widetilde{\signal{u}}\in\real^M_{++}$ by placing zeros at the same locations of the zero components in $\signal{x}$. (For example, if $\signal{x}=(1,0,2)$ and $\widetilde{\signal{u}}=(3,4)$, then $\widetilde{\signal{x}}=(1,2)$ and $\signal{u}=(3,0,4)$.) In such a case, we verify that $\mathrm{diag}(\signal{x})T_{\|\cdot\|}(\signal{u}) = \rho((\signal{D}\widetilde{T})_{\|\cdot\|_\star})\signal{u}$, which, from the definition of the spectral radius in \refeq{eq.nl_radius} implies that $\ngn(\signal{x})=\rho(\mathrm{diag}(\signal{x})T_{\|\cdot\|})\ge \rho((\signal{D}\widetilde{T})_{\|\cdot\|_\star}) >0$ as claimed.
	
	(iii) Fix $(\signal{x},\signal{y})\in\real_+^N\times \real_+^N$ with $\signal{x}\le\signal{y}$ arbitrarily. Coordinate-wise nonnegativity of $\tnorm$ implies that $(\forall\signal{u}\in\real_+^N)~ \mathrm{diag}(\signal{x})\tnorm(\signal{u})\le \mathrm{diag}(\signal{y})\tnorm(\signal{u}).$ Let $\signal{z}\in\real_+^{N}\backslash\{\signal{0}\}$ be the eigenvector of $\mathrm{diag}(\signal{x})\tnorm$ corresponding to the spectral radius $\rho(\mathrm{diag}(\signal{x})\tnorm)$. As a result, $\rho(\mathrm{diag}(\signal{x})\tnorm)\signal{z} = \mathrm{diag}(\signal{x})\tnorm(\signal{z})\le \mathrm{diag}(\signal{y})\tnorm(\signal{z})$, which implies $\ngn(\signal{x})=\rho(\mathrm{diag}(\signal{x})\tnorm)\le \rho(\mathrm{diag}(\signal{y})\tnorm)=\ngn(\signal{y})$ in light of Fact~\ref{fact.ineq_spec}(iii) in the Supplemental Material.
	
	(iv) For a matrix $\signal{M}\in\real^{N\times N}$, denote by $\|\signal{M}\|$ the matrix operator norm induced by the vector norm $\|\cdot\|$. Let $\mathcal{P}$ be the compact set given by $\mathcal{P}:=\{\signal{x}\in\real_+^N\mid \|\signal{x}\|=1 \}$. Since $\tnorm$ is a general interference mapping, and hence continuous according to our definition (see the discussion below Definition~\ref{def.sif}), the extreme-value theorem shows that there exists $B\in\real_+$ such that $B=\sup_{\signal{p}\in\mathcal{P}}\|\tnorm(\signal{p})\|$, and the supremum is attained at some $\signal{p}^\star\in\mathcal{P}$. As a result, if a sequence $(\signal{x}_n)_{n\in\Natural}$ in $\real_+^N$ converges to $\signal{x}^\star\in\real_+^N$, we can  use the standard property $(\forall\signal{x}\in\real^N)(\forall\signal{M}\in\real^{N\times N})\|\signal{Mx}\|\le \|\signal{M}\|~\|\signal{x}\|$ of matrix operator norms to deduce
	\begin{multline*}
		\lim_{n\to\infty}\sup_{\signal{p}\in\mathcal{P}}\|\mathrm{diag}(\signal{x}_n)\tnorm(\signal{p})-\mathrm{diag}(\signal{x}^\star)\tnorm(\signal{p})\| \\ 
		\le \lim_{n\to\infty}\sup_{\signal{p}\in\mathcal{P}}\|\mathrm{diag}(\signal{x}_n)-\mathrm{diag}(\signal{x}^\star)\|~\|\tnorm(\signal{p})\|\\
		\le \lim_{n\to\infty}\|\mathrm{diag}(\signal{x}_n)-\mathrm{diag}(\signal{x}^\star)\| B = 0.
	\end{multline*}
	We can now apply Fact~\ref{fact.continuity} in the Supplemental Material to deduce:
	\begin{multline*}
		\lim_{n\to\infty}\ngn(\signal{x}_n)=\lim_{n\to\infty}\rho(\mathrm{diag}(\signal{x}_n)\tnorm) \\= \rho(\mathrm{diag}(\signal{x}^\star)\tnorm)=\ngn(\signal{x}^\star),
	\end{multline*}
	which shows that $\ngn$ is continuous on $\real^N_+$.
\end{proof}

The function $\ngn:\real_+^N\to\real_+$ defined in Proposition~\ref{proposition.qnorm} satisfies almost all the properties of a norm on the nonnegative cone $\real_{+}^N$. The only exception is the triangle inequality. This missing property is automatically satisfied on the nonnegative cone if $\ngn$ is known to be convex only in the positive cone $\real_{++}^N$, in which case we can extend $\ngn$ to the entire domain $\real^N$ while keeping the properties of a norm:

\begin{proposition}
	\label{prop.norm}
	Let $\|\cdot\|$ be a monotone norm on $\real^N$, $T:\real_+^N\to\real_{++}^N$ and $\tnorm:\real_+^N\to\real_{+}^N$ satisfy hypothesis ($T$,$\tnorm$), and  $\ngn:\real_+^N\to\real_+$ be as defined in Proposition~\ref{proposition.qnorm}. Assume that $\ngn$ is a convex function in $\real_{++}^N$. Then $\ngn$ is also convex in $\real_+^N$. Furthermore, for $C_\tnorm:=\{\signal{x}\in\real^N_{+}~\mid \ngn(\signal{x})\le 1\}$ and $S_\tnorm=\mathrm{conv}(C_\tnorm\cup -C_\tnorm)$, where $-C_\tnorm:=\{\signal{x}\in\real^N\mid -\signal{x}\in C_\tnorm\}$ and $\mathrm{conv}(\cdot)$ is the convex hull of a set, the Minkowski or gauge functional of $S_\tnorm$, defined by
	\begin{align}
		\label{eq.minkowski}
		(\forall\signal{x}\in\real^N)~\|\signal{x}\|_\tnorm := \inf\{\gamma>0\mid (1/\gamma)\signal{x}~\in S_\tnorm\},
	\end{align}
	is a monotone norm satisfying $(\forall\signal{x}\in\real_{+}^N)~\|\signal{x}\|_\tnorm=\ngn(\signal{x}).$
\end{proposition}
\begin{proof}
	Here we use extensively the properties of $\ngn$ in Proposition~\ref{proposition.qnorm}(i)-(iv). The proof that convexity on the interior of the domain of an everywhere continuous function implies convexity on the whole domain is standard, but we show the main steps for completeness. Assume that $(\signal{x}_n)_{n\in\Natural}$ and $(\signal{y}_n)_{n\in\Natural}$ are two sequences in $\real_{++}^N$ converging in the normed vector space $(\real^N,\|\cdot\|)$ to, respectively,  the vectors  $\signal{x}\in \real_{+}^N$ and $\signal{y}\in \real_{+}^N$.  Given $\alpha\in~]0,1[$, convexity of $\ngn$ on $\real_{++}^N$ implies that $(\forall n\in\Natural)~\ngn(\alpha\signal{x}_n+(1-\alpha)\signal{y}_n)\le \alpha \ngn(\signal{x}_n)+(1-\alpha)\ngn(\signal{y}_n)$. Passing to the limit as $n\to\infty$ and recalling that $\ngn$ is continuous on $\real^N_+$, we deduce $\ngn(\alpha\signal{x}+(1-\alpha)\signal{y})\le \alpha \ngn(\signal{x})+(1-\alpha)\ngn(\signal{y})$, which establishes convexity of $\ngn$ on $\real_{+}^N$. In particular, setting $\alpha$ to $\alpha=1/2$ and using homogeneity of $\ngn$, we deduce $\alpha\ngn(\signal{x}+\signal{y})\le \alpha \ngn(\signal{x})+\alpha \ngn(\signal{y})$, which proves that $\ngn$ satisfies the triangle inequality on $\real_+^N$. Therefore, $\ngn$ satisfies all properties of a norm when restricted to vectors in the closed cone $\real_+^N$, so the set $C_\tnorm:=\{\signal{x}\in\real^N_{+}~\mid \ngn(\signal{x})\le 1\}$ is a compact convex set with nonempty interior. Furthermore, $C_\tnorm$ is downward comprehensive on $\real_{+}^N$ as a consequence of Proposition~\ref{proposition.qnorm}(iii). All the properties of $C_\tnorm$ we have proved enable us to apply \cite[Proposition~2]{renatomaxmin} to conclude that the Minkowski functional in \refeq{eq.minkowski} is indeed a monotone norm, and it satisfies 
	\begin{align}
		\label{eq.setc}
		C_\tnorm=\{\signal{x}\in\real_+^N\mid \|\signal{x}\|_\tnorm\le 1\}=\{\signal{x}\in\real_+^N\mid \ngn(\signal{x})\le 1\};
	\end{align}
	i.e., the norm $\|\cdot\|_\tnorm$ and the function $\ngn$ generate the same level set at level $1$. We now prove that $\|\cdot\|_\tnorm$ and $\ngn$ are pointwise equal on $\real_+^N$. Proposition \ref{proposition.qnorm}(ii) shows that $0=\|\signal{0}\|_\tnorm=g_\tnorm(\signal{0})$, so we only need to consider the case  $\signal{x}\in\real_+^N\backslash\{\signal{0}\}$. We deduce from $\ngn(\signal{x})>0$ (Proposition~\ref{proposition.qnorm}(ii)) and positive homogeneity of $\ngn$ that $1= \ngn((1/\ngn(\signal{x}))\signal{x})$, so  $(1/\ngn(\signal{x}))\signal{x}\in C_\tnorm$ from the second characterization of the set $C_\tnorm$ in \refeq{eq.setc}. From the first characterization of $C_\tnorm$ in \refeq{eq.setc} and $(1/\ngn(\signal{x}))\signal{x}\in C_\tnorm$, we have $\|(1/\ngn(\signal{x}))\signal{x}\|_\tnorm\le 1$, and thus 
	\begin{align}
		\label{eq.ineq1}
	\|\signal{x}\|_\tnorm \le \ngn(\signal{x}).
	\end{align}
	 Similarly, from the first characterization of $C_\tnorm$ in \refeq{eq.setc}, we verify that $(1/\|\signal{x}\|_\tnorm)\signal{x} \in C_\tnorm$. Therefore, using the second characterization of $C_\tnorm$ in \refeq{eq.setc}, we deduce $g_\tnorm((1/\|\signal{x}\|_\tnorm)\signal{x})\le 1$, which implies \begin{align}
	 	\label{eq.ineq2}
	 	\|\signal{x}\|_\tnorm \ge \ngn(\signal{x})
	 \end{align}
	 because $g_\tnorm$ is positively homogeneous. Combining \refeq{eq.ineq1} and \refeq{eq.ineq2}, we conclude the proof that $(\forall\signal{x}\in\real_+^N)~\|\signal{x}\|_\tnorm=g_\tnorm(\signal{x})$.
\end{proof}

The result in Proposition \ref{prop.norm} provides a foundation for a natural definition of a norm as follows:

\begin{definition}
	\label{def.norm_inducing} 
	Let $T:\real_+^N\to\real_{++}^N$ and $\tnorm:\real_+^N\to\real_{+}^N$ satisfy hypothesis ($T$,$\tnorm$) for a given monotone norm $\|\cdot\|$. If the function $\real_+^N\to\real_+:\signal{x}\mapsto \rho(\mathrm{diag}(\signal{x})\tnorm)$ is convex on $\real_{++}^N$, we say that $\tnorm$ is norm-inducing, and we call the norm $\|\cdot\|_\tnorm$ in \refeq{eq.minkowski} the norm induced by $\tnorm$.
\end{definition}

Our primary focus is on constructing norm-inducing general interference mappings derived from an explicitly given standard interference mapping and a monotone norm. However, in certain instances, it is straightforward to extract a suitable standard interference mapping and a monotone norm from a general interference mapping $G$ to verify that $G$ is norm-inducing, as illustrated in the next example, which is common in wireless networks -- see, for example, the set in \refeq{eq.slinear} in Sect.~\ref{sect.nopower}. 

\begin{example}
	\label{example.M}
	Consider the general interference mapping \linebreak[4] $G:\real_+^N\to\real_{+}^N:\signal{x}\mapsto \signal{Mx}$, where $\signal{M}\in\real_{++}^{N\times N}$ is an inverse Z-matrix. Then $G$ is norm-inducing, and the induced norm satisfies $(\forall\signal{x}\in\real_+^N)~\|\signal{x}\|_G=\rho(\diagvec{x}\signal{M})$.
\end{example}
\begin{proof}		
	Standard Perron-Frobenius theory shows that the spectral radius $\rho(G)$ is positive, and there exists a positive left eigenvector $\signal{l}\in\real^N_{++}$ and a positive right eigenvector $\signal{r}\in\real_{++}^N$ associated with $\rho(G)$. Choose $\epsilon>0$ such that $\signal{M}_1 = \signal{M}-\epsilon~\signal{lr}^t\in\real_{+}^{N\times N}$ (a scalar $\epsilon$ with this property always exists because $\signal{M}$ is positive). The mapping $T:\real_+^N\to\real_{++}^N:\signal{x}\mapsto \signal{M}_1\signal{x}+\epsilon\signal{l}$ is a standard interference mapping because it is positive and concave \cite[Proposition~1]{cavalcante2016}. Using the monotone norm $(\forall\signal{x}\in\real^N)~\|\signal{x}\|_\signal{r}:=\signal{r}^t|\signal{x}|$,  we have $(\forall\signal{x}\in\real_+^N)~T_{\|\cdot\|_\signal{r}}(\signal{x})= G(\signal{x})=\signal{Mx}=\signal{M}_1\signal{x}+\epsilon\signal{l}\|\signal{x}\|_\signal{r}$, so $T$ and $T_{\|\cdot\|_\signal{r}}$ satisfy hypothesis $(T,T_{\|\cdot\|_\signal{r}})$. Therefore,  we are in the setting of Definition~\ref{def.norm_inducing} because $g_G:\real_{+}^N\to\real_+^N:\signal{x}\mapsto \rho(\mathrm{diag}(\signal{x})G)=\rho(\mathrm{diag}(\signal{x})\signal{M})$ is convex on $\real_+^N$ as a direct consequence of Fact~\ref{fact.friedland} in the Supplemental Material, so $G$ is norm-inducing as claimed, and the induced norm satisfies $(\forall\signal{x}\in\real_+^N)\|\signal{x}\|_G=\rho(\diagvec{x}\signal{M})$.
\end{proof}

\begin{remark}
	\label{remark.triangular}
	The function $G:\real^N_+\to\real_+:\signal{x}\mapsto \rho(\mathrm{diag}(\signal{x})\signal{M})$ is also convex if $\signal{M}\in\real_{+}^{N\times N}$ is either (i) irreducible and symmetric positive semidefinite \cite[Theorem~1.51]{slawomir09}, or (ii) triangular (upper or lower). In the triangular case, for a given matrix $\signal{M}\in\real^{N\times N}_+$, the function $G:\real_{+}^N\to\real_+:\signal{x}=(x_1,\dots,x_N)\mapsto\rho(\mathrm{diag}(\signal{x})\signal{M})=\max_{n\in\{1,\ldots,N\}} x_n \signal{M}_{n,n}$ (NOTE: $\signal{M}_{n,n}$ denotes the $n$th diagonal entry of $\signal{M}$) is the maximum of convex (linear) functions, hence convex.  Note that nonnegative matrices in the classes above need not be inverse Z-matrices or norm-inducing. We do not pursue these cases here, but they may be useful for establishing convexity of utility regions beyond the scope of the present study.
\end{remark}

The next result establishes a sufficient condition ensuring that a mapping $\tnorm$ of the type considered in Example~\ref{example.illustration} is norm-inducing.

\begin{example}
	\label{example.induced_norm}
	Let $\signal{M}\in\real_{+}^{N\times N}$ be a nonnegative matrix, $\signal{u}\in \real_{++}^N$ a positive vector, and $\|\cdot\|$ a polyhedral monotone norm. As a result, there exist $K\in\Natural$ vectors  $\signal{a}_1,\ldots,\signal{a}_K$ in $\real^N_{+}\backslash\{\signal{0}\}$ satisfying $(\forall\signal{x}\in\real^N) \|\signal{x}\|=\max_{n\in\{1,\ldots,K\}}\signal{a}_n^t|\signal{x}|$. Define $T:\real_+^N\to\real_{++}^N:\signal{x}\mapsto \signal{Mx}+\signal{u}.$
		The mapping $T$ is positive and concave, so it is also a standard interference mapping \cite[Proposition~1]{cavalcante2016}. By construction, $T$ and $\tnorm:\real_+^N\to\real_{+}^N:\signal{x}\mapsto \signal{Mx}+\signal{u}\|\signal{x}\|$ satisfy hypothesis ($T$,$\tnorm$).  If the matrices $(\signal{M}+\signal{u}\signal{a}_n^t)_{n\in\{1,\ldots,K\}}$ are inverse Z-matrices, then $\tnorm$ is norm-inducing, and the induced norm $\|\cdot\|_\tnorm$ satisfies 
	\begin{multline}
		\label{eq.Gnorm}
		(\forall\signal{x}\in\real_+^N)~\|\signal{x}\|_\tnorm=\rho(\mathrm{diag}(\signal{x})\tnorm)\\=\max_{n\in\{1,\ldots,K\}}\rho(\mathrm{diag}(\signal{x})(\signal{M}+\signal{u}\signal{a}^t_n)).
	\end{multline}

\end{example}
\begin{proof}
	Fix $\signal{x}\in\real^N_{++}$, and, for each $n\in\{1,\ldots,K\}$, define  $\signal{M}_n:=\signal{M}+\signal{u}\signal{a}^t_n$.  Let $\signal{z}$ be the eigenvector corresponding to the spectral radius $\rho(\mathrm{diag}(\signal{x})\tnorm)$, and note that $\signal{z}\in\real_{++}^N$ and $\rho(\mathrm{diag}(\signal{x})\tnorm)>0$ as a consequence of Proposition~\ref{proposition.cond_eig}. Therefore, we have
	\begin{multline*}
		(\forall n\in\{1,\ldots,K\})~
		\rho(\mathrm{diag}(\signal{x})\tnorm) \signal{z} = \mathrm{diag}(\signal{x})\tnorm(\signal{z}) \\ = \mathrm{diag}(\signal{x})\signal{M}_{n^\star}\signal{z}  
		\ge \mathrm{diag}(\signal{x})\signal{M}_n\signal{z},
	\end{multline*}
	where $n^\star\in\argmax_{n\in\{1,\ldots,K\}}\signal{a}_n^t\signal{z}$.
	We conclude from the above inequalities and Fact~\ref{fact.ineq_spec}(i)-(ii) in the Supplemental Material that 
	\begin{align}
		\label{eq.ineqa}
		\ngn(\signal{x})=\rho(\mathrm{diag}(\signal{x})\tnorm) = \max_{n\in\{1,\ldots,K\}} \rho(\mathrm{diag}(\signal{x})\signal{M}_n).
	\end{align}
	For each $n\in\{1,\ldots,K\}$, the matrix $\signal{M}_n$ is an inverse Z-matrix by assumption, so Fact~\ref{fact.friedland} in the Supplemental Material and \refeq{eq.ineqa} show that the function $\ngn$, restricted to $\real_{++}^N$, is the pointwise maximum of convex functions. As a result, $\ngn$ is convex on $\real_{++}^N$ \cite[Proposition~8.16]{baus17}. Therefore, in light of Proposition \ref{prop.norm}, $\tnorm$ is norm-inducing and the scalar $\ngn(\signal{x})$ in $\refeq{eq.ineqa}$ corresponds to the value of the induced norm $\|\cdot\|_\tnorm$ evaluated at the chosen $\signal{x}\in\real_+^N$.
\end{proof}

\section{Applications to wireless networks}
\label{sect.model}

\subsection{Interference models}
We now focus on the SINR as the utility function for the uplink of a wireless system with $N \in \mathbb{N}$ single-antenna users, which are components of the set $\mathcal{N}:=\{1,\ldots,N\}$. However, some results we discuss below can also be extended to the downlink using duality arguments \cite{boche2002,schubert2024duality,miretti2024ul}\cite[Sect.~4.3.2]{massivemimobook}. To avoid technology specific models, we refrain from specifying the number of access points, the number of antennas at these access points, or the particular wireless architecture in use (e.g., cellular massive MIMO, cell-less, etc.). For many of these technologies, given the uplink power allocation $\signal{p}^\star=(p_1^\star,\ldots,p_N^\star)\in\real_{++}^N$, the SINR $s_n$ for each user $n\in\mathcal{N}$ can be expressed in a unified form as follows:
\begin{align}\label{eq.sinr}
	(\forall n\in\mathcal{N})~ s_n=\sup_{\signal{y}\in\mathcal{Y}_n}\dfrac{b_n(\signal{y})~p_n^\star}{\signal{c}_n(\signal{y})^t\signal{p}^\star+\sigma_n(\signal{y})},
\end{align}
where $\mathcal{Y}_n$ is typically a set of random or deterministic vectors and, for every $\signal{y}\in\mathcal{Y}_n$, the parameters $b_n(\signal{y}) \in \real_{++}$, $\signal{c}_n(\signal{y}) \in \real_{+}^{N}$, and $\sigma_n(\signal{y})\in\real_{++}$ are system specific -- see Examples~\ref{example.cell-less} and \ref{example.cellular} below. We assume that, for each $n\in\mathcal{N}$, the function $\mathcal{Y}_n\to\real_{++}:\signal{y}\mapsto\sigma_n(\signal{y})$ is bounded away from zero (i.e., $(\exists \delta>0)(\forall n\in\mathcal{N})(\forall \signal{y}\in\mathcal{Y}_n)~\sigma_n(\signal{y})\ge\delta$), which is a mild assumption that typically holds in networks with noise. To avoid technical digressions of little practical relevance, we also assume that the supremum in \refeq{eq.sinr} is always attained and, for every $n\in\mathcal{N}$, we have  $s_n\in\real_{++}$ if $\signal{p}^\star\in\real_{++}^N$. As Example~\ref{example.cell-less} below shows, in current studies,  $\mathcal{Y}_n$ is often (though not always -- see Example~\ref{example.cellular}) the set of available beamformers for user $n\in\mathcal{N}$, and the remaining parameters, which are functions of the beamformers, are constructed based on information about the channels, thermal noise, and the quality of channel estimates. Particular cases of this model can be found in, for instance,  cell-less networks  \cite[Ch.~3, Ch.~7]{demir2021}\cite{miretti2022closed,miretti2025two,schubert2024duality,miretti2024sum,chafaa2025}, massive MIMO systems  \cite[Ch.5.3.2]{marzetta16}, \cite[Theorem7.1]{massivemimobook}, and conventional cellular networks \cite[Ch.~4]{slawomir09}.  The next two examples illustrate how to specialize \refeq{eq.sinr} to model modern cell-less networks and classical cellular networks with optimal cell-site selection.  

\begin{example}
	\label{example.cell-less}
	(Two-time scale receive beamforming in cell-less systems \cite{miretti2025two,miretti2022closed,chafaa2025})
	In the uplink of a cell-less network with single-antenna users, let the channel of user $n\in\mathcal{N}$ across all $M\in\Natural$ access points, each equipped with $L\in\Natural$ antennas, be a $\mathbb{C}^{ML}$-valued random vector $\signal{h}_n$. Assume a fixed (digital) receive beamforming strategy; i.e., for each user $n\in\mathcal{N}$, the set $\mathcal{Y}_n$ is a singleton containing a $\mathbb{C}^{ML}$-valued random vector $\signal{v}_n$. For instance, if the access points use as their local receive beamformers the  traditional maximum ratio combining strategy, the random vector $\signal{v}_n$ can be interpreted as an estimate of the instantaneous channel $\signal{h}_n$. Nevertheless, we emphasize that other strategies (such as zero-forcing receive beamforming) are allowed, as long as the receive beamformer remains unchanged if the power allocation changes. This assumption is adopted in \cite{miretti2022closed,chafaa2025} and also in the numerous references cited therein. To compute the well-known use-and-then-forget (UatF) ergodic capacity inner bound \cite{marzetta16} for user $n\in\mathcal{N}$, we use the following SINR expression for a given power allocation $\signal{p}^\star=(p_1^\star,\ldots,p_N^\star)\in\real_{++}^N$:
	\begin{align}
		\label{eq.sinrcl}
		s_n=\dfrac{p^\star_n |E[\signal{h}_n^H\signal{v}_n]|^2}{p^\star_n V(\signal{h}^H_n\signal{v}_n)+\sum \limits_{k\in\mathcal{N}\backslash\{n\}}p^\star_k~E[|\signal{h}_k^H\signal{v}_n|^2]+E[\|\signal{v}_n\|_2^2]},
	\end{align}
	where $E(\cdot)$ and $V(\cdot)$ denote, respectively, the expectation and the variance of a random variable. Using the notation in \refeq{eq.sinr}, for $n=1$, we verify that $\mathcal{Y}_1=\{\signal{v}_1\}$, $b_1(\signal{v}_1)=|E[\signal{h}_1^H\signal{v}_1]|^2$, $\sigma_1(\signal{v}_1)=E[\|\signal{v}_1\|_2^2]$, and $$\signal{c}_1(\signal{v}_1) = (V(\signal{h}^H_1\signal{v}_1), E[|\signal{h}_2^H\signal{v}_1|^2],\ldots,E[|\signal{h}_N^H\signal{v}_1|^2]).$$ An essential component to note in this model is the self-interference term $V(\signal{h}^H_1\signal{v}_1)$ (which is also related to the notion of \emph{beamforming gain uncertainty} in the massive {MIMO} literature \cite{marzetta16}). This term is crucial in our analysis of the shape of the SINR and achievable rate regions, as explained later in Sect.~\ref{sect.impact}. 
\end{example}

\begin{example}
	\label{example.cellular}
	(Cellular systems with optimal cell-site selection \cite{hanly1995algorithm}) 	Consider the uplink of a cellular network with $Y\in\Natural$ base stations indexed by the set $\mathcal{Y}=\{1,\ldots,Y\}$, and let the  $N\in\Natural$ users be indexed by the set $\mathcal{N}=\{1,\ldots,N\}$. Let $b_n(y)\in\real_{++}$ denote the channel gain between user $n\in\mathcal{N}$ and base station $y\in\mathcal{Y}_n$, where $\mathcal{Y}_n\subset\mathcal{Y}$ represents the set of base stations to which user $n\in\mathcal{N}$ is allowed to connect. For simplicity, assume that the noise level $\sigma\in\real_{++}$ is identical for all links.
	
	Define the interference channel vector for user $n\in\mathcal{N}$ at base station $y\in\mathcal{Y}_n$ by  $\signal{c}_n(y)=[c_{1,n}(y),\cdots,c_{N,n}(y)]^t\in\real_{++}^N$, where 
	$c_{k,n}(y)=b_k(y)$ if $k\neq n$ and $0$ otherwise. If $\signal{p}^\star=[p^\star_1,\ldots,p^\star_N]^t\in\real_{++}^N$ is the uplink transmit power vector for all users, and each user $n\in\mathcal{N}$ connects to its best serving base station, then the resulting SINR of user $n\in\mathcal{N}$ is given by
	\begin{align*}
		s_n=\max_{y\in\mathcal{Y}_n}\dfrac{b_n(y)p_n^\star}{\signal{c}_n(y)^t\signal{p}^\star+\sigma},
	\end{align*}
	which is a particular case of \refeq{eq.sinr}.
\end{example}

In the remainder of this study, Example~\ref{example.cell-less} will be used extensively to illustrate the main implications of our results. Nevertheless, a few remarks regarding channel acquisition are in order.

\begin{remark}
As noted in Example~\ref{example.cell-less}, the beamformers in \refeq{eq.sinrcl} may be designed based on \emph{channel estimates} rather than the true channels, so we naturally account for channel estimation errors. In the sequel, however, we assume that the channel estimation quality -- and the resources allocated to channel estimation -- are independent of the particular set of users being served. Relaxing this assumption would substantially complicate the analysis. In particular, regardless of whether the achievable-rate region is convex or nonconvex, determining whether users benefit from being grouped would then require an explicit model of the estimation procedure (e.g., how pilot overhead and the resulting estimation accuracy change when users are partitioned into smaller or larger groups).
\end{remark}

We can also establish a connection between \refeq{eq.sinr} and the discussion in Sect.~\ref{sect.overview}. More specifically, the function $t_n$ in \refeq{eq.utility} for user $n\in\mathcal{N}$ is given by 
\begin{align}
	\label{eq.fn}
	t_n:\real_+^N\to\real_{++}:\signal{p}\mapsto \inf_{\signal{y}\in\mathcal{Y}_n} \dfrac{\signal{c}_n(\signal{y})^t\signal{p}+\sigma_n(\signal{y})}{b_n(\signal{y})}
\end{align}
(which we assume to be continuous) and $s_n$ in \refeq{eq.sinr} is the value taken by the utility $U_n:\real_+^N\to\real_+$ in \refeq{eq.utility} of user $n\in\mathcal{N}$ for a given power vector $\signal{p}^\star\in\real^N_{++}$; i.e., $U_n(\signal{p}^\star)=p^\star_n/t_n(\signal{p}^\star)=s_n$. The function $t_n$ in \refeq{eq.fn} is a standard interference function for every $n\in\mathcal{N}$ because, for each $\signal{y}\in\mathcal{Y}_n$, the function  $\signal{p}\mapsto (\signal{c}_n(\signal{y})^t\signal{p}+\sigma_n(\signal{y}))/b_n(\signal{y})$ is affine (hence concave), so $t_n$ is concave because it is the pointwise infimum of concave functions \cite[Proposition~8.16]{baus17}.\footnote{For every $n\in\mathcal{N}$, the restriction  $t_n:\real^N_{++}\to\real_+$ of $t_n$ to the domain $\real_{++}^N$ admits a continuous, concavity-preserving extension to the domain $\real_+^N$ \cite[Theorem 10.3]{rock70} that is also bounded away from zero. Therefore, we can always redefine $t_n$ to be this continuous extension and assume continuity of $t_n$ on $\real_{+}^N$.} Furthermore, $t_n$ is positive everywhere on the nonnegative cone $\real_{+}^N$ by assumption (which in particular implies that $(\forall n\in\mathcal{N})\inf_{\signal{y}\in\mathcal{Y}_n}\sigma_n(\signal{y})/b_n(\signal{y})>0$), so, being positive and concave, $t_n$ is a standard interference function \cite[Proposition~1]{cavalcante2016}.

To impose power constraints in the system model, we define the set of valid power levels by: 
\begin{align}
	\label{eq.C}
	\mathcal{P}:=\{\signal{p}\in\real_+^N\mid \|\signal{p}\|\le 1\},
\end{align}
where $\|\cdot\|$ is a monotone norm. Power constraints of this type are fairly general because any set with nonempty interior that is compact, convex, and downward comprehensive in $\real_+^N$ can be expressed as in \refeq{eq.C} for some monotone norm $\|\cdot\|$ \cite[Proposition~2]{renatomaxmin}. Readers may notice that set $\mathcal{P}$ in \refeq{eq.C} differs from the set $\mathcal{P}$ already introduced in Sect.~\ref{sect.overview} because the scalar $p_\mathrm{max}\in\real_{++}$ is omitted in \refeq{eq.C}. However, this omission does not result in any loss of generality. Given any monotone norm $\|\cdot\|_\star$, the constraint $\|\signal{p}\|_\star\le p_\mathrm{max}$ can be equivalently expressed as the set $\mathcal{P}$ in \refeq{eq.C} by considering the scaled monotone norm $(\forall\signal{p}\in\real^N)~ \|\signal{p}\|:=(1/p_\mathrm{max})\|\signal{p}\|_\star$, which satisfies $(\forall\signal{p}\in\real^N)~ \|\signal{p}\|\le 1 \Leftrightarrow \|\signal{p}\|_\star \le {p_\mathrm{max}}$.

To rewrite \refeq{eq.sinr} in a form suitable for analysis via fixed point theory, we isolate $p_n^\star\in\real_{++}$ in the numerator of the fraction inside the supremum for each $n\in\mathcal{N}$. Using the coordinate functions defined in \refeq{eq.fn}, we obtain  $(\forall n\in\mathcal{N})~ p_n^\star=s_n t_n(\signal{p}^\star),$ which, in vector form, can be equivalently expressed as  
\begin{align}
	\label{eq.matrix_form}
	\signal{p}^\star=\diagvec{s}T(\signal{p}^\star),
\end{align}
where $T:\real_+^N\to\real_{++}^N:\signal{p}\mapsto \left[t_1(\signal{p}),\cdots, t_N(\signal{p})\right]^t$ and  $\signal{s}=(s_1,\ldots,s_N)\in\real_{++}^N$ is the vector of SINR values.

To relate the above model to the results derived in Sect.~\ref{sect.convexity} -- and, in particular, Definition~\ref{def.compatible} -- consider the monotone norm $\|\cdot\|$ in the definition of the set $\mathcal{P}$ in \refeq{eq.C}, and define
 \begin{align}
	\label{eq.G} \tnorm:\real_{+}^N\to\real_+^N:\signal{p}\mapsto [\tn{1}(\signal{p}),\ldots,\tn{N}(\signal{p})]^t,
\end{align} where $\tn{n}$ for $n\in\mathcal{N}$ is given by
\begin{align*}
	\tn{n}:\real_+^N\to\real_{++}:\signal{p}\mapsto \inf_{\signal{y}\in\mathcal{Y}_n}\left( \dfrac{\signal{c}_n(\signal{y})^t\signal{p}}{b_n(\signal{y})} + \dfrac{\sigma_n(\signal{y})}{b_n(\signal{y})} \|\signal{p}\| \right)
\end{align*}
(which we also assume to be continuous). We verify that $T$ and $\tnorm$ defined as above satisfy hypothesis $(T,\tnorm)$ for the monotone norm $\|\cdot\|$ used in \refeq{eq.C}.

\subsection{Achievable rate and SINR regions in the absence of power constraints}
\label{sect.nopower}
To study the feasible SINR and rate regions of the wireless model introduced above, let us start by ignoring the presence of the power constraint $\mathcal{P}$ to relax key assumptions made in \cite[Ch.~5]{slawomir09}, particularly the linearity of the interference models and the symmetry of the interference matrices (cf. \cite[Ch.~5.4.4]{slawomir09}). 

The equality in \refeq{eq.matrix_form} proves that, if the SINR requirements described by the positive vector $\signal{s}$ are achievable by the power vector $\signal{p}^\star\in\real_{++}^N$, then
\begin{align}
	\label{eq.fixed_point}
	\signal{p}^\star\in\mathrm{Fix}(\diagvec{s}T):=\{\signal{p}\in\real_{++}^N~|~\diagvec{s}T(\signal{p})=\signal{p} \},
\end{align}
where $T:\real_{+}^N\to\real_{++}^N$ is the standard interference mapping defined below \refeq{eq.matrix_form}. Therefore, the set $\mathrm{Fix}(\diagvec{s}T)$ is either a singleton or the empty set \cite{yates95}, and we verify that the SINR requirements represented by the vector $\signal{s}$ are achievable by a power vector $\signal{p}^\star\in\real_{++}^N$ if and only if $\mathrm{Fix}(\diagvec{s}T)\neq\emptyset$, and this power vector is unique if it exists \cite{yates95}. In compact form, the set of positive SINR values that the network can provide to the users can be expressed as $\mathcal{S}=\{\signal{s} \in\real_{++}^N \mid \mathrm{Fix}(\mathrm{diag}(\signal{s})T)\neq\emptyset\}.$  This approach based on fixed point theory has been introduced to the wireless literature in \cite{yates95}, and it has since seen significant advancements in both theory and practice \cite{nuzman07,martin11,cavalcante2019,piotrowski2022,miretti2024ul}.

In the absence of power constraints, we can use the concept of asymptotic mappings (Definition~\ref{def.am}) to obtain a second characterization of the set $\mathcal{S}$ of positive, achievable SINR levels $\signal{s}=(s_1,\ldots,s_N)\in\real_{++}^N$. More precisely, given $\signal{s}\in\real^N_{++}$, continuity of the mapping $\signal{x}\mapsto \mathrm{diag}(\signal{s})\signal{x}$ yields 
\begin{multline*}
	(\forall\signal{s}\in\real_{++}^N)(\forall\signal{p}\in\real^N)\\ (\mathrm{diag}(\signal{s})T)_\infty (\signal{p})=
	\lim_{h\to\infty}(1/h)\mathrm{diag}(\signal{s})T(h\signal{p}) \\= \mathrm{diag}(\signal{s})\lim_{h\to\infty}(1/h)T(h\signal{p})  = \mathrm{diag}(\signal{s})T_\infty(\signal{p}),
\end{multline*}
which, together with Fact~\ref{fact.tinf} in the Supplemental Material, leads to the following alternative characterization of $\mathcal{S}$:
\begin{multline}
	\label{eq.s_no_constraints}
	\mathcal{S}=\{\signal{s} \in\real_{++}^N \mid \mathrm{Fix}(\mathrm{diag}(\signal{s})T)\neq\emptyset\} \\ = \{\signal{s} \in\real_{++}^N \mid \rho(\mathrm{diag}(\signal{s}){T}_\infty)<1\}.
\end{multline}

\begin{example}
	\label{example.linear}
	Consider the model and notation in Example~\ref{example.cell-less}. To simplify notation, we drop the arguments $(\signal{v}_n)_{n\in\mathcal{N}}$ because the sets $(\mathcal{Y}_n=\{\signal{v}_n\})_{n\in\mathcal{N}}$ are assumed to be singletons. For instance, $\signal{c}_n(\signal{v}_n)$, $b_n(\signal{v}_n)$, and $\sigma_n(\signal{v}_n)$ for each $n\in\mathcal{N}$ is simply denoted by $\signal{c}_n$, $b_n$, and $\sigma_n$, respectively. As a result, by defining $(b_1,\ldots,b_N) =: \vec{b} \in \real_{++}^{N}$, $(\vec{c}_1,\ldots,\vec{c}_N) =: \vec{C}\in \real_{+}^{N\times N}$, and $(\sigma_1,\ldots, \sigma_N)=: \vec{\sigma}\in\real_{++}^N$, the mapping $T:\real^N_+\to\real_{++}^N$ in the fixed point equation $\signal{p}^\star=\diagvec{s}T(\signal{p}^\star)$ in \refeq{eq.matrix_form} simplifies to 
	$T:\real_+^N\to\real^N_{++}:\signal{p}\mapsto \signal{Mp}+\signal{u},$ where $\signal{M} := \mathrm{diag}(\vec{b})^{-1}\vec{C}^t$ is referred to as the \emph{interference matrix},  and $\vec{u} := \mathrm{diag}(\signal{b})^{-1}\signal{\sigma}.$ For this mapping $T$, we readily verify that its corresponding asymptotic mapping $T_\infty$ (see Definition~\ref{def.am}) is given by $T_\infty:\real^N_+\to\real_+^N:\signal{p}\mapsto \signal{Mp},$ and thus the achievable SINR region $\mathcal{S}$ in \refeq{eq.s_no_constraints} simplifies to 
	\begin{align}
		\label{eq.slinear}
		\mathcal{S}:= \{\signal{s} \in\real_{++}^N \mid \rho(\mathrm{diag}(\signal{s})\signal{M})<1\},
	\end{align}
	which is the characterization of the achievable SINR region discussed in \cite[Ch.~5]{slawomir09} for another particular case of \refeq{eq.sinr}. Therefore, the characterization of the set $\mathcal{S}$ in \refeq{eq.s_no_constraints} -- which does not assume linear asymptotic mappings $T_\infty$ and includes the linear mapping $T_\infty$ in this Example as a special case -- can be seen as a generalization of the results in \cite[Ch.~5]{slawomir09} to possibly nonlinear systems. In particular, \refeq{eq.s_no_constraints} covers the model in Example~\ref{example.cellular}, and it can also accommodate optimal beamforming strategies, such as those studied in \cite{miretti2025two}. 
\end{example}

For the affine interference model in the previous example, we can also obtain an easily verifiable, sufficient condition for convexity of $\mathcal{S}$ in \refeq{eq.slinear}:\footnote{The results in Remark~\ref{remark.triangular} may be useful for establishing convexity for interference models beyond those considered in this study.}  
\begin{remark}
	\label{remark.nopc}
	If the interference matrix $\signal{M}$ in Example~\ref{example.linear} is an inverse Z-matrix, then, by Fact~\ref{fact.friedland} in the Supplemental Material, the function $\mathcal{D}_{+}^{N\times\N}\to\real_+:\signal{D}\mapsto\rho(\signal{DM})$ is convex. Therefore, the SINR region $\mathcal{S}$ in \refeq{eq.slinear} is convex because this set is the intersection of the level set of a convex function with the convex cone $\real_{++}^N$. 
\end{remark}
More generally, if we use the standard nonlinear bijection $\mathrm{rate} = \log(1+\mathrm{SINR})$ (NOTE: here $\log$ represents the natural logarithm, so the rate is expressed in nats per symbol), in which case we say that we are \emph{treating interference as noise (single-user decoder)}, we can express the relationship between rates and SINR of all users as $\signal{D}=e^\signal{R}-\signal{I}$, where $\signal{D}=\mathrm{diag}(s_1,\ldots,s_N)$ is the diagonal matrix of SINR values, $\signal{R}=\mathrm{diag}(r_1,\ldots,r_N)$ is the corresponding diagonal matrix of rates, and $\signal{I}$ denotes the identity matrix. Therefore, recalling \refeq{eq.s_no_constraints}, we can express the set of achievable rates $\mathcal{R}$ as the following set
\begin{align}
	\label{eq.R}
	\mathcal{R}:=\{\signal{r}\in \real_{+}^N \mid \rho((e^\diagvec{r}-\signal{I})T_\infty)<1\}.
\end{align} 
(Here we consider all achievable rates in the full domain $\real^N_+$ to avoid technical digressions in the proof of the next proposition). To establish convexity of $\mathcal{R}$, we use the next result:

\begin{proposition}
	\label{proposition.Rconvexity}
	Assume that the function $\real_+^N\to\real_+:\signal{s}\mapsto\rho(\diagvec{s}T_\infty)$ is convex for a given standard interference mapping $T:\real_+^N\to\real_{++}^N$. Then the function $$f:\real_+^N\to\real_+:\signal{r}\mapsto \rho((e^\diagvec{r}-\signal{I})T_\infty)$$ is also convex.
\end{proposition}
\begin{proof}
	Recall that $[\cdot]_+$ is the operator that sets the negative components of a vector or matrix to zero. Consider the function
	\begin{align*}
		\begin{array}{rcl}
			h:\real^N&\to&\real_+ \\ (x_1,\ldots,x_N) &\mapsto& \rho([\mathrm{diag}(x_1,\ldots,x_N)]_+T_\infty),
		\end{array}
	\end{align*}
	The results in Fact~\ref{fact.ineq_spec} in the Supplemental Material imply that, for any two general interference mappings $G_1:\real_+^N\to\real_+^N$ and $G_2:\real_+^N\to\real_+^N$ such that $(\forall\signal{x}\in\real_+^N)~G_1(\signal{x})\le G_2(\signal{x})$, we have $\rho(G_1)\le\rho(G_2)$. Therefore, $h$ is monotone in the sense that $(\forall\signal{x}\in\real^N) (\forall\signal{y}\in\real^N)~ \signal{x} \le \signal{y} \Rightarrow h(\signal{x})\le h(\signal{y})$. To prove that $h$ is also convex, choose two vectors $(\signal{x}_1,\signal{x}_2)\in\real^{N}\times \real^{N}$ and a scalar $\alpha\in~]0,1[$ arbitrarily. We establish convexity of $h$ as follows:
	\begin{multline}
		\label{eq.chainineq}
		h(\alpha\signal{x}_1+(1-\alpha)\signal{x}_2) \le h(\alpha [\signal{x}_1]_+ + (1-\alpha)[\signal{x}_2]_+) \\ \le \alpha h([\signal{x}_1]_+) + (1-\alpha) h([\signal{x}_2]_+)  = \alpha h(\signal{x}_1) + (1-\alpha) h(\signal{x}_2),
	\end{multline}
	where the first inequality is a consequence of monotonicity of $h$, and the second inequality follows from the assumption of convexity of $\signal{s}\mapsto\rho(\diagvec{s}T_\infty)$ on $\real_{+}^N$. Therefore, we deduce from \cite[p. 86]{boyd}\cite[Theorem 2.1.3(vi)]{zalinescu2002} that the function $h\circ g:\real^N\to\real_+$, where $$g:\real^N\to\real^N:(x_1,\ldots,x_N)\mapsto (e^{x_1}-1,\ldots,e^{x_N}-1),$$ is also convex  because $h\circ g$ is the composition of the monotonic convex function $h$ and the mapping $g$, which is convex in each coordinate. Convexity of $h\circ g$ implies convexity of the function $f$ because $f$ is the restriction of the convex function $h\circ g$ to the nonnegative orthant.
\end{proof}

The discussion above establishes that convexity of the mapping $\real_{++}^N\to\real_+:\signal{s}\mapsto\rho(\mathrm{diag}(\signal{s})T_\infty)$ implies the convexity of the set $\mathcal{S}$ in \refeq{eq.s_no_constraints} because $\mathcal{S}$ is the level set of a convex function. As an immediate implication of this result and Proposition~\ref{proposition.Rconvexity}, we deduce that the set $\mathcal{R}$ in \refeq{eq.R} is also convex. Moreover, Remark~\ref{remark.nopc} provides a simple criterion for establishing convexity in special cases. In more detail, for affine interference models such as those described in Example~\ref{example.linear}, the sets $\mathcal{S}$ and $\mathcal{R}$ are convex whenever the interference matrix $\signal{M}$ is an inverse Z-matrix. For later reference, we formalize these results below.

\begin{Cor}
	\label{cor.setr}
	If the asymptotic mapping $T_\infty:\real_+^N\to\real_+^N$ is such that the function $\real_{++}^N\to\real_+:\signal{s}\mapsto\rho(\mathrm{diag}(\signal{s})T_\infty)$ is convex, then the sets $\mathcal{S}$ and $\mathcal{R}$ in \refeq{eq.s_no_constraints} and \refeq{eq.R}, respectively, are convex. In particular, for the interference model in Example~\ref{example.linear}, the sets $\mathcal{S}$ and $\mathcal{R}$ are convex if the interference matrix $\signal{M}$ in the definition of the linear asymptotic mapping  $T_\infty:\real_{+}^N\to\real_+^N:\signal{p}\mapsto\signal{Mp}$ is an inverse Z-matrix.
\end{Cor}

\subsection{Achievable rate and SINR regions with power constraints}
\label{sect.withpower}

In the previous section, we discussed existing results proving that there exists at most one power vector $\signal{p}^\star\in\real_{++}^N$ able to provide a given SINR requirement $\signal{s}=(s_1,\ldots,s_N)\in\real_{++}^N$, and, if it exists, this power vector is precisely the fixed point of the mapping $\diagvec{s}T$ -- see \refeq{eq.fixed_point}. Therefore, to describe the set $\mathcal{S}_\mathcal{P}$  of positive achievable SINR values with power constraints, we have to remove from the set $\mathcal{S}$ in \refeq{eq.s_no_constraints} all points $\signal{s}$ corresponding to feasible SINR levels obtained with a power vector $\signal{p}^\star\notin \mathcal{P}$, where $\mathcal{P}$ is defined in \refeq{eq.C}. More precisely, verifying whether $\mathrm{Fix}(\diagvec{s}T)$ is nonempty is insufficient. If $\mathrm{Fix}(\diagvec{s}T)\neq\emptyset$, we also need to determine whether the unique fixed point $\signal{p}^\star\in \mathrm{Fix}(\diagvec{s}T)$, which is the only power allocation able to provide the SINR requirement represented by the vector $\signal{s}$, satisfies $\signal{p}^\star\in\mathcal{P}$. From Fact~\ref{fact.feasibility} in the Supplemental Material we know that $\signal{p}^\star\in \mathrm{Fix}(\diagvec{s}T)\neq\emptyset$ and $\signal{p}^\star\in\mathcal{P}$ if and only if the conditional eigenvalue $\lambda$ associated with $\diagvec{s}T$ and $\|\cdot\|$ satisfies $\lambda\le 1$, and Proposition~\ref{proposition.cond_eig} shows that $\lambda\le 1$ if and only if $\rho(\diagvec{s}\tnorm)\le 1$, where $\tnorm$ is the general interference mapping defined in \refeq{eq.G}. As a result, the set of positive SINR values achievable under the power constraint $\mathcal{P}$, excluding the points on the boundary  $\real_+^N\backslash\real^N_{++}$ of the nonnegative cone $\real_+^N$, can be expressed in the following compact form:
\begin{multline}
	\label{eq.sp}
	\mathcal{S}_\mathcal{P}:=\{\signal{s}\in\real_{++}^N\mid \rho(\diagvec{s}\tnorm)\le 1\}.
\end{multline}
In turn, the set of achievable rates obtained by treating interference as noise is given by:
\begin{multline}
	\label{eq.rp}
	\mathcal{R}_\mathcal{P}:=\{\signal{r}\in\real_{++}^N\mid \rho((e^\diagvec{r}-\signal{I})\tnorm)\le 1\}.
\end{multline}

We can now obtain a parametrization of the weak Pareto boundary of the SINR and rate regions in terms of the spectral radius of the general interference mapping $\tnorm$ in \refeq{eq.G}: 

\begin{proposition}
	\label{proposition.wp}
	Let $\signal{s}\in\real_{++}^N$ be a given SINR vector. Then $\signal{s}$ lies on the weak Pareto boundary of $\mathcal{S}_\mathcal{P}$ in \refeq{eq.sp} if and only if $\rho(\diagvec{s}\tnorm) = 1$. Similarly, assume that $\signal{r}\in\real_{++}^N$ is a given rate vector. Then $\signal{r}$ lies on the weak Pareto boundary of $\mathcal{R}_\mathcal{P}$ in \refeq{eq.rp} if and only if $\rho((e^{\diagvec{r}}-\signal{I})\tnorm)= 1$.
\end{proposition}
\begin{proof}
	In this proof, we only consider the set $\mathcal{S}_\mathcal{P}$ because the proof for the set $\mathcal{R}_\mathcal{P}$ follows from the standard bijection $\text{rate}=\log(1+\text{SINR})$. 	From Proposition~\ref{proposition.qnorm}(ii) we know that $0<\alpha:=\rho(\diagvec{s}\tnorm)$ because $\signal{s}\neq\signal{0}$. For the sake of contradiction, assume that $\alpha < 1$ and $\signal{s}$ is on the weak Pareto boundary of $\mathcal{S}_\mathcal{P}$. Positive homogeneity of the spectral radius yields $1 = (1/\alpha) \rho(\mathrm{diag}(\signal{s})\tnorm) = \rho(\mathrm{diag}((1/\alpha) \signal{s})\tnorm)$, which shows that $(1/\alpha) \signal{s}$ is a feasible SINR vector, which contradicts that $\signal{s}$ is on the weak Pareto boundary because $(1/\alpha) \signal{s}>\signal{s}$; i.e., we can increase all components of $\signal{s}$ while remaining in the set $\mathcal{S}_\mathcal{P}$. To prove the converse, assume that $\rho(\mathrm{diag}(\signal{s})\tnorm)=1$ for a given $\signal{s}\in\real_{++}^N$, and let $\signal{\epsilon}\in\real_{++}^N$ be arbitrary. Proposition~\ref{proposition.cond_eig} shows that the eigenvector $\signal{x}$ of the mapping $\signal{u}\mapsto \diagvec{s}\tnorm(\signal{u})$ associated with the spectral radius $\rho(\mathrm{diag}(\signal{s})\tnorm)$ is positive. Therefore, from the definition of (nonlinear) eigenvectors and eigenvalues, we have 
	\begin{align*}
		\signal{0}<\signal{x} = \mathrm{diag}(\signal{s})\tnorm(\signal{x}) \Leftrightarrow \signal{0}<[\mathrm{diag}(\signal{s})]^{-1}\signal{x} = \tnorm(\signal{x}),
	\end{align*}
	and thus $\signal{0}<\signal{x} < \mathrm{diag}(\signal{s}+\signal{\epsilon})\tnorm(\signal{x})$. This last inequality proves that there exists a scalar $\lambda > 1$ satisfying $\lambda\signal{x} \le \mathrm{diag}(\signal{s}+\signal{\epsilon})\tnorm(\signal{x})$, which implies $\rho(\mathrm{diag}(\signal{s}+\signal{\epsilon})\tnorm)>1$ as a consequence of Fact~\ref{fact.ineq_spec}(iii) in the Supplemental Material. Therefore, we have  $\signal{s}+\signal{\epsilon}\notin \mathcal{S}_\mathcal{P}$ for every $\signal{\epsilon}\in\real_{++}^N$, so $\signal{s}$ lies on the weak Pareto boundary of $\mathcal{S}_\mathcal{P}$.
\end{proof}

If $\tnorm$ is norm-inducing in the sense of Definition \ref{def.norm_inducing}, with induced norm denoted by $\|\cdot\|_\tnorm$, which satisfies $(\forall\signal{s}\in\real_{+}^N)~\|\signal{s}\|_\tnorm= \rho(\diagvec{s}\tnorm)$ as shown in Proposition~\ref{prop.norm}, the closure $\bsp$ of the set $\mathcal{S}_\mathcal{P}$ is given by 
\begin{align}
	\label{eq.csp}
	\bsp:=\{\signal{s}\in\real_{+}^N\mid \|\signal{s}\|_\tnorm\le 1\},
\end{align}
which is a compact convex set. In this case, the closure $\brp$ of the set $\mathcal{R}_\mathcal{P}$ in \refeq{eq.rp} can be equivalently expressed as:
\begin{multline}
	\label{eq.brp}
	\brp:=\{\signal{r}=(r_1,\ldots,r_N)\in\real_{+}^N\mid \\ \|(e^{r_1}-1,\ldots,e^{r_N}-1)\|_\tnorm\le 1\},
\end{multline}
which is also compact and convex as proved below: 
\begin{proposition}
	\label{proposition.norm_inducing_G}
	Assume that $\tnorm$ in \refeq{eq.G} is norm-inducing. Then the set $\brp$ of achievable rates is a compact convex set.
\end{proposition}
\begin{proof}
	We follow the same procedure as in the proof of Proposition~\ref{proposition.Rconvexity}, highlighting only the necessary modifications for the sake of brevity. Consider the function
	\begin{align*}
		\begin{array}{rcl}
			h:\real^N&\to&\real_+ \\ (x_1,\ldots,x_N) &\mapsto& \|[(x_1,\ldots,x_N)]_+\|_\tnorm.
		\end{array}
	\end{align*}
	
	Choose two vectors $(\signal{x}_1,\signal{x}_2)\in\real^N\times \real^N$ and a scalar $\alpha\in~]0,1[$, and recall that the norm $\|\cdot\|_\tnorm$ is monotone as consequence of Proposition~\ref{proposition.qnorm}(iii) and Proposition~\ref{prop.norm}. Convexity and monotonicity of $\|\cdot\|_\tnorm$ imply convexity of $h$:
	\begin{multline*}
		h(\alpha\signal{x}_1+(1-\alpha)\signal{x}_2) \le h(\alpha [\signal{x}_1]_+ + (1-\alpha)[\signal{x}_2]_+) \\ \le \alpha h([\signal{x}_1]_+) + (1-\alpha) h([\signal{x}_2]_+) \\ = \alpha h(\signal{x}_1) + (1-\alpha) h(\signal{x}_2).
	\end{multline*}
	Now, follow the exact same steps below \refeq{eq.chainineq} to conclude that the set $\brp$ is convex. Boundedness of $\signal{r}^\star\in\brp$ follows immediately from boundedness of closed balls $\{\signal{x}\in\real^N \mid \|\signal{x}\|_\tnorm\le 1\}$, nonnegativity of vectors in $\brp$, and the fact that $x\to\infty$ implies $e^x\to\infty$. Therefore, being closed and bounded in the standard (finite-dimensional) Euclidean space, $\brp$ is compact in this space.
\end{proof}

Below we show a sufficient condition for $\tnorm$ to be norm-inducing, so that the sets $\mathcal{S}_\mathcal{P}$ and $\mathcal{R}_\mathcal{P}$ are convex. This condition is immediate from Example~\ref{example.induced_norm} and Proposition~\ref{proposition.norm_inducing_G}.

\begin{Cor}
	\label{cor.withpc}
	For the wireless model in Examples~\ref{example.cell-less} and ~\ref{example.linear}, let $\signal{a}_1,\ldots,\signal{a}_K$ be vectors in $\real_{+}^N\backslash\{\signal{0}\}$ such that the polyhedral monotone norm $\|\cdot\|$ in the definition of the power constraint $\mathcal{P}$ in \refeq{eq.C} can be expressed as $\|\signal{x}\|=\max\limits_{n\in\{1,\ldots,K\}}\signal{a}_n^t|\signal{x}|$. Assume that $\signal{M}+\signal{u}\signal{a}_n^t$ is an inverse Z-matrix for each $n\in\{1,\ldots,K\}$. Then the sets $\mathcal{S}_\mathcal{P}$ and $\mathcal{R}_\mathcal{P}$ are convex. The closure of these sets, shown in \refeq{eq.csp} and \refeq{eq.brp}, uses the norm $\|\cdot\|_\tnorm$ with the property in \refeq{eq.Gnorm}. 
\end{Cor}

\begin{remark} \label{remark.simplicity} With $K=N$ and $\signal{a}_1,\ldots,\signal{a}_K$ as the standard basis vectors,  Corollary~\ref{cor.withpc} could be partially deduced from the findings in \cite{friedland2008maximizing}. Even in this particular case of the results we present, our approach offers some significant advantages. First, convexity of the SINR and rate regions can be identified even in scenarios with strong interference and high SNR, which is a regime that cannot be covered by results in previous studies. For example, the hypothetical interference matrix $\signal{M}=\left[\begin{matrix}2&10 \\ 10^{-1} & 1\end{matrix}\right]$ has a component (in the first row, second column) corresponding to strong interference compared to all other components of the matrix. Nevertheless, $\signal{M}+\signal{ua}_n^t$ is an inverse Z-matrix for every $n\in\{1,\ldots,K\}$ if the components of all vectors $(\signal{a}_n)_{n\in\{1,\ldots,N\}}$ are sufficiently small, which can be guaranteed by making the transmit power sufficiently large. Therefore, we are neither in the low SNR regime nor in the low interference regime, yet the conditions of Corollary~\ref{cor.withpc} hold. Second, our approach focuses on the spectral radius of a single nonlinear mapping (see \refeq{eq.sp}), which not only streamlines the notation but also accommodates more flexible power constraints. More generally, the characterization of the feasible rate and SINR regions through the spectral radius of the nonlinear mappings in \refeq{eq.sp} and \refeq{eq.rp} does not require affine interference models or convex polytopes as the power constraints, so our characterization is inaccessible to previous results that do not use the concept of nonlinear spectral radius.
\end{remark}

\section{Discussion}
\label{sect.implication}

To highlight further the relevance of the preceding results for network designers, except for Sect.~\ref{subsect.sum}, we focus on the interference model in Examples \ref{example.cell-less} and \ref{example.linear}. Notably, that interference model is very common in the massive/extremely large MIMO and cell-less literature (see \cite{miretti2022closed,chen2018achievable,chafaa2025} and the references therein). In the discussion that follows, we extensively use Fact~\ref{fact.zm} in the Supplemental Material, which establishes that an invertible interference matrix $\signal{M}\in\real_{+}^{N\times N}$, being nonnegative, is an inverse Z-matrix if and only if it is an inverse M-matrix. This relationship allows us to utilize the insights from the previous sections, along with the comprehensive body of research on these matrix classes \cite{johnson2011}, to explain interference patterns observed in common network models for resource allocation. 

\subsection{Impact of self-interference on the convexity of the feasible rate and SINR regions}
\label{sect.impact}

The diagonal entries of the interference matrix $\signal{M}\in\real_{+}^{N\times N}$ defined in Example~\ref{example.linear} correspond to the so-called \emph{self-interference}, which in many models in the massive MIMO and cell-less literature is caused by the lack of perfect channel state information. For Remark~\ref{remark.nopc} to be applicable, self-interference must exist in all communication links $n\in\mathcal{N}$. The reason is that a necessary (but not sufficient) condition for $\signal{M}\in\real_+^{N\times N}$ to be an inverse Z-matrix  (or, equivalently, an inverse M-matrix, as mentioned above) is that all diagonal entries must be positive \cite[Theorem~2.9]{johnson2011}. 

In contrast, Corollary~\ref{cor.withpc} may be applicable even in the absence of self-interference in some links because of the presence of the additive matrices $(\signal{u}\signal{a}_n^t)_{n\in\{1,\ldots,K\}}$. However, their presence does not  guarantee that $(\signal{M}+\signal{u}\signal{a}_n^t)_{n\in\{1,\ldots,K\}}$ are inverse Z-matrices, which is the main assumption required in Corollary~\ref{cor.withpc}. Nevertheless, as self-interference increases in every link, all feasible SINR and rate regions eventually become convex under the practical assumption that all links cause interference to each other, as proved below:

\begin{Cor}
	\label{cor.selfint} Assume that the interference matrix $\signal{M}\in\real_+^{N\times N}$ in Example~\ref{example.linear} has positive off-diagonal entries, and it is replaced by $\signal{S}+\signal{M}$, where $\signal{S}\in\mathcal{D}_{+}^{N\times N}$ models the addition of self-interference to the system. Then there exists a nonnegative diagonal matrix $\signal{S}_\mathrm{min}\in\mathcal{D}_{+}^{N\times N}$ such that the sets $\mathcal{S}$, $\mathcal{S}_\mathcal{P}$, $\mathcal{R}$, and $\brp$ are convex for every $\signal{S} \ge \signal{S}_\mathrm{min}$.
\end{Cor}
\begin{proof}
	Fact~\ref{fact.diag} in the Supplemental Material shows that there exist scalars $(\alpha_1,\ldots,\alpha_{K+1})\in\real_{+}^{K+1}$ such that $(\alpha_n\signal{I}+\signal{M}+\signal{u}\signal{a}_n^t)_{n\in\{1,\ldots,K\}}$ and $\alpha_{K+1}\signal{I}+\signal{M}$ are inverse M-matrices. It then follows from Fact~\ref{fact.dsum} in the Supplemental Material that, for any $\signal{S}\in\mathcal{D}_{+}^{N\times N}$ satisfying $\signal{S}\ge\signal{S}_\mathrm{min}:=\bar{\alpha}\signal{I}$, where  $\bar{\alpha}:=\max\limits_{n\in\{1,\ldots,K+1\}} \alpha_n$, the matrices $\signal{M}+\signal{S}$, $\signal{M}+\signal{S}+\signal{u}\signal{a}_1^t$, ...,  $\signal{M}+\signal{S}+\signal{u}\signal{a}_K^t$ are inverse M-matrices, hence also inverse Z-matrices because they are nonnegative (Fact~\ref{fact.zm} in the Supplemental Material). A direct application of Remark \ref{remark.nopc}, Corollary \ref{cor.setr}, and Corollary \ref{cor.withpc} completes the proof.
\end{proof}

\subsection{Users offering compatible channels}
\label{sect.Zcompat}
Assume that two single-antenna users need to be served by a single base station equipped with $L\in\Natural$ antennas (we consider the general case soon). In the massive MIMO literature, a key concept indicating whether both users can be served simultaneously while maintaining good link performance is \emph{favorable} or \emph{approximately favorable propagation} \cite[Ch.~7]{marzetta16}. More precisely, we say that the users offer \emph{favorable propagation} if their channels $(\signal{h}_1,\signal{h}_2)\in\mathbb{C}^L\times \mathbb{C}^L$ are such that $\signal{h}_1^H\signal{h}_2=0$, or that they offer \emph{approximately favorable propagation} if $|\signal{h}_1^H\signal{h}_2|$ is sufficiently small (possibly in a probabilistic sense for models such as those in Example~\ref{example.cell-less}). It is often recommended to avoid scheduling users on the same resources when favorable propagation is absent \cite[p.~154]{marzetta16}, and this principle can be trivially generalized to systems with multiple users and base stations by considering the channels for each link pair. However, this system design recommendation overlooks key aspects of real systems, including, among others: (i) it does not take into account the presence of beamformers and the available transmit power;  (ii) it ignores self-interference; (iii) many results proving approximate favorable propagation consider asymptotic regimes by letting the number of antennas diverge to infinity, but systems with a small number of antennas are common; (iv) users cannot be simply dropped -- they need to be scheduled at some time slot considering some fairness criterion; (v) and it lacks mathematical guarantees that users would not achieve better performance by being divided into smaller groups and served at different time slots.

The results in Sect.~\ref{sect.model} provide a framework for defining a notion of user compatibility that overcomes the above drawbacks. Specifically, if the achievable rate region is convex, as discussed in Sect.~\ref{sect.overview}, we have a mathematical certificate that the network performance cannot be improved by splitting users into smaller groups and by serving these groups at different times, irrespective of the point on the weak Pareto boundary being considered (the fairness criterion). Therefore, the following definition is justified:

\begin{definition}
	\label{def.ch_compatible}
	Consider the system model in Examples \ref{example.cell-less} and \ref{example.linear}. We say that the users offer \emph{Z-compatible channels without power constraints} if the interference matrix $\signal{M}$ is an inverse Z-matrix. Likewise,  we say that the users offer Z-compatible channels with the power constraint $\mathcal{P}$ in \refeq{eq.C} if the matrices  $(\signal{M}+\signal{ua}_n^t)_{n\in\{1,\ldots,K\}}$ are inverse Z-matrices.
\end{definition}

Corollaries \ref{cor.setr} and ~\ref{cor.withpc} show  that Z-compatibility guarantees convexity of the SINR and rate regions.

\subsection{Numerical aspects}
A simple means to verify channel Z-compatibility across a group of users, in the sense of Definition~\ref{def.ch_compatible}, is to invert the matrix $\signal{M}$ or the matrices $(\signal{M}+\signal{u}\signal{a}_n^t)_{n\in\{1,\ldots,K\}}$ (if they are invertible) and check whether the resulting matrices are Z-matrices (i.e., the off-diagonal entries are nonpositive). However, this direct approach based on matrix inverses becomes computationally expensive and numerically unstable in large systems. These issues can be mitigated by leveraging the vast literature on the identification of inverse Z-matrices or inverse M-matrices without actually computing any inverses, which is a well-studied problem with applications across different scientific and engineering domains \cite{johnson2011}. Below, we explore some of these established results to reduce the computational burden of identifying users with Z-compatible or incompatible channels. For brevity, we focus on the interference matrix $\signal{M}$ because the same principles apply to the matrices $(\signal{M}+\signal{u}\signal{a}_n^t)_{n\in\{1,\ldots,K\}}$ without any changes.

If the interference matrix $\signal{M}$ is of dimension $2\times 2$ (i.e., we try to produce a certificate of channel Z-compatibility for two users), inverting such a small matrix may not pose serious challenges. Nevertheless, it is more efficient and numerically stable to apply the result stated in Fact~\ref{fact.2x2} in the Supplemental Material. More specifically, an interference matrix $\signal{M}\in\real_+^{2\times 2}$ is an inverse Z-matrix if and only if its determinant $\mathrm{det}(\signal{M})$ satisfies $\mathrm{det}(\signal{M})>0$. Unfortunately, this simple inequality is only a necessary condition for larger matrices \cite[Theorem~1.2]{johnson2011}. Nevertheless, all principal minors of inverse Z-matrices are positive \cite[Corollary~2.3.3]{johnson2011}, which lead us to a simple rule for identifying pairs of users \emph{potentially} leading to a nonconvex SINR or rate region:

\begin{remark}
	\label{remark.zmatrix}
	Consider the interference model in Examples \ref{example.cell-less} and \ref{example.linear}. Let $\signal{M}\in\real_+^{N\times N}$ be the interference matrix of a network with $N$ users. Given two users $(i,j)\in\mathcal{N}\times\mathcal{N}$, let $\signal{M}[i,j]\in\real^{2\times 2}_+$ be the reduced interference matrix obtained from $\signal{M}$ by keeping only the $i$th and $j$th columns and the $i$th and $j$th rows. If $\mathrm{det}(\signal{M}[i,j])\le 0$, then the full matrix  $\signal{M}$ is not an inverse Z-matrix.
\end{remark}

\subsection{A counter-example to the conjecture in \cite[Sect.~5.4.4]{slawomir09}\cite{stanczak2007convexity}}
\label{subsect.conjecture}

Given a matrix $\signal{M}\in\real_+^{N\times N}$, the studies in \cite[Sect.~5.4.4]{slawomir09}\cite{stanczak2007convexity} conjecture that convexity of the function $l_\signal{M}:\real^N_+\to\real_+:\signal{x}\mapsto \rho(\diagvec{\signal{x}}\signal{M})$ is linked to the positive semidefiniteness of the symmetrized interference matrix $\signal{M}+\signal{M}^t$. If this conjecture were true, it would provide a simple sufficient condition for verifying the convexity of the SINR region $\mathcal{S}$ for the system model in Example~\ref{example.cell-less}. We now show that positive semidefiniteness of $\signal{M}+\signal{M}^t$ is neither necessary nor sufficient for convexity. More specifically, if the interference matrix $\signal{M}\in\real^{N\times N}_+$ is an inverse Z-matrix, in which case the function $l_\signal{M}$ is convex, then there exists a positive diagonal matrix $\signal{D}\in\real^{N\times N}_+$ such that the symmetrized interference matrix $\signal{DM}+\signal{M}^t\signal{D}$ is positive definite \cite[Theorem~1.4]{johnson2011}, but $\signal{D}$ is not necessarily the identity matrix. Therefore, positive definiteness of $\signal{M}+\signal{M}^t$ is not a necessary condition for convexity of $l_\signal{M}$. To show that this property is not sufficient, take $\signal{x}_1=[0.5,~0.1,~1]^t$, $\signal{x}_2=[0.5,~0.5,~0.5]^t$, $\alpha=0.9$, and
\begin{align*}
	\signal{M}=\left[\begin{matrix}
		11  & 10   & 1 \\
		1 & 11 & 10 \\
		10    & 10   & 10     
	\end{matrix}\right].
\end{align*}
We can verify that $\signal{M}+\signal{M}^t$ is positive definite and that $l_{\signal{M}}$ is not quasi-convex (and, hence, not convex) because  $l_{\signal{M}}(\alpha\signal{x}_1+(1-\alpha)\signal{x}_2)> \max\{ l_{\signal{M}}(\signal{x}_1), l_{\signal{M}}(\signal{x}_2)\}$.

\subsection{The sum-rate maximization problem}
\label{subsect.sum}
In this final subsection of Sect.~\ref{sect.implication}, we turn the attention to the interference model in \refeq{eq.sinr} (or, equivalently \refeq{eq.fn}) in its full generality, with the models presented in Examples \ref{example.cell-less} and \ref{example.linear} serving as particular cases. In more detail, many scheduling and resource allocation schemes rely heavily on the sum-rate maximization problem, which is known to be NP-hard in general \cite{luo2008dynamic}. However, if $\tnorm$ in \refeq{eq.G} is norm-inducing, the closure $ \brp=\{\signal{r}\in\real_+^N\mid \rho((e^{\diagvec{r}}-\signal{I})\tnorm)\le 1\}$ of the set of achievable rates $\mathcal{R}_\mathcal{P}$ in \refeq{eq.rp} is convex as shown in Proposition~\ref{proposition.norm_inducing_G}, and $\brp$ can be equivalently expressed using a monotone norm as shown in \refeq{eq.brp}. In this case, (weighted) sum-rate maximization can be posed as the following convex optimization problem, which opens up the possibility for deriving efficient solvers that provably converge to global optima (particular solvers will be considered elsewhere):
\begin{align}
	\label{eq.sumrate}
	\begin{array}{rl}
		\mathrm{max.} & \signal{w}^t\signal{r} \\
		\mathrm{s.t.} & \signal{r}\in\brp = \{\signal{r}\in\real_{+}^N\mid \rho((e^\diagvec{r}-\signal{I})\tnorm)\le 1\},		
	\end{array}
\end{align}
where $\signal{w}\in\real_{++}^N$ are the desired weights. Once a solution $\signal{r}^\star$ to \refeq{eq.sumrate} is obtained, assuming that $\signal{r}^\star\in\real_{++}^N$, we can recover the optimal power allocation $\signal{p}^\star\in\real_{++}^N$ by computing the fixed point of the standard interference mapping $\signal{D}^\star T$ in \refeq{eq.fixed_point}, where $\signal{D}^\star=e^\diagvec{r^\star}-\signal{I}$ is the matrix of optimal SINR levels corresponding to the optimal rates $\signal{r}^\star$. In particular, the uniquely existing fixed point $\signal{p}^\star\in\mathrm{Fix}(\signal{D}^\star T)$ can be computed with the standard fixed point iterations \cite{yates95,nuzman07} or with acceleration methods \cite{cavalcante2016}. If $T$ is also a positive concave mapping, convergence is guaranteed to be geometric in the standard Euclidean space \cite{piotrowski2022}. If any components of $\signal{r}^\star$ are zero, we can simply remove the corresponding users from the system and consider a fixed point problem with a reduced dimension. 

It is worth emphasizing that the class of sum-rate maximization problems that can be reformulated as convex optimization problems without change in variables does not necessarily require $\tnorm$ to be norm-inducing, and the results from Sect.~\ref{sect.model} offer analytical support for a straightforward guideline for designing sum-rate maximization algorithms:

\begin{remark}
	\label{remark.design}
	We have shown that convexity of the function $\signal{s}\mapsto\rho(\diagvec{s}\tnorm)$ implies convexity of the set $\bsp$ and also convexity of the set $\brp$, as established in Proposition~\ref{proposition.norm_inducing_G}. However,  $\brp$ can be convex even if neither the set $\bsp$ nor the function $\signal{s}\mapsto\rho(\diagvec{s}\tnorm)$ is convex. This fact suggests to use the achievable rates directly as the optimization variables in sum-rate maximization problems, as done in \refeq{eq.sumrate}. The alternative approach of using SINR values $\signal{s}=(s_1,\ldots,s_N)$ as the optimization variables (or the power levels providing the required SINR values) and the sum-rate function $\sum_{n\in\mathcal{N}}\log(1+s_n)$ as the cost function may require the introduction of the set $\bsp$ as a constraint, which is less likely to be convex compared to the set $\brp$. Therefore, we may lose the opportunity to exploit a potential hidden convexity of the sum-rate maximization problem, increasing the risk of iterative solvers converging to local rather than global optima, if convergence can be guaranteed.
\end{remark}

\section{Simulations}
To illustrate the results from the previous section using the model in  Examples \ref{example.cell-less} and \ref{example.linear}, we examine the uplink of a small cell-less network consisting of four access points, each equipped with two antennas, uniformly distributed within a 100m × 100m square area. Three single-antenna users are positioned uniformly at random in this area, with each user connected to the two access points providing the strongest channels. The small number of users and access points is chosen deliberately because plotting the feasible SINR and rate regions is impossible if the number of users exceeds three. Nevertheless, the results in the previous section are valid for systems of any dimension. For the beamformers, we employ maximum ratio combining. Concretely, in the SINR expression in \refeq{eq.sinrcl}, the aggregate beamformer vector $\signal{v}_n$ (obtained by stacking the beamformers of every access point) for each user $n\in\mathcal{N}=\{1,2,3\}$ is a scaled version of the aggregated channel vector $\signal{h}_n$, with zeros placed in the coordinates corresponding to the access points to which the user is not connected. Expectations are replaced by empirical averages from 100 realizations of the random variables, and all channel parameters and noise samples are constructed exactly as described in \cite[Sect. III-D]{miretti2022closed}, so we omit the details for brevity. 
For the monotone norm $\|\cdot\|$ in the power constraint $\mathcal{P}$ in \refeq{eq.C}, we use the $l_\infty$ norm, scaled such that the maximum transmit power per user is limited to $p_{\max} =0.2$ W; i.e., we use the polyhedral monotone norm $(\forall\signal{p}\in\real^N)~ \|\signal{p}\|:=(1/p_{\max})~\|\signal{p}\|_\infty$. Equivalently, the vectors $(\signal{a}_n)_{n\in\{1,2,3\}}$ in Corollary~\ref{cor.withpc} are the standard basis vectors appropriately scaled so that the users do not exceed the maximum transmit power.

To display the feasible rate and SINR regions, we proceed as follows. For each user $n\in \mathcal{N}=\{1,2,3\}$, the initial transmit power $p_n$ is a sample drawn uniformly at random from the interval $]0,1]$. The transmit power vector $\signal{p}=[p_1,~p_2,~p_3]$ is then normalized so that the power constraint is satisfied with equality, thus we have $\|\signal{p}\|=1$. This normalization ensures that the power vector produces rates and SINR levels lying on the weak Pareto boundary, as proved in \cite{cavalcante2023}. We generate 1,000 such power vectors using this method, and the corresponding SINR and rate values are shown as point clouds in 3D plots.

Fig.~\ref{fig.convex} illustrates results for a sample user placement scenario where all matrices $(\signal{M}+\signal{u}\signal{a}^t_n)_{n\in\{1,2,3\}}$ described in Example \ref{example.cell-less} are inverse Z-matrices. These figures indicate convexity of the feasible rate and SINR regions, which is in agreement with the results in Corollary~\ref{cor.withpc}. 

\begin{figure}
	\centering
	\begin{subfigure}[b]{0.45\columnwidth}
		\centering
		\includegraphics[width=\textwidth, trim=2cm 0cm 2cm 2cm, clip]{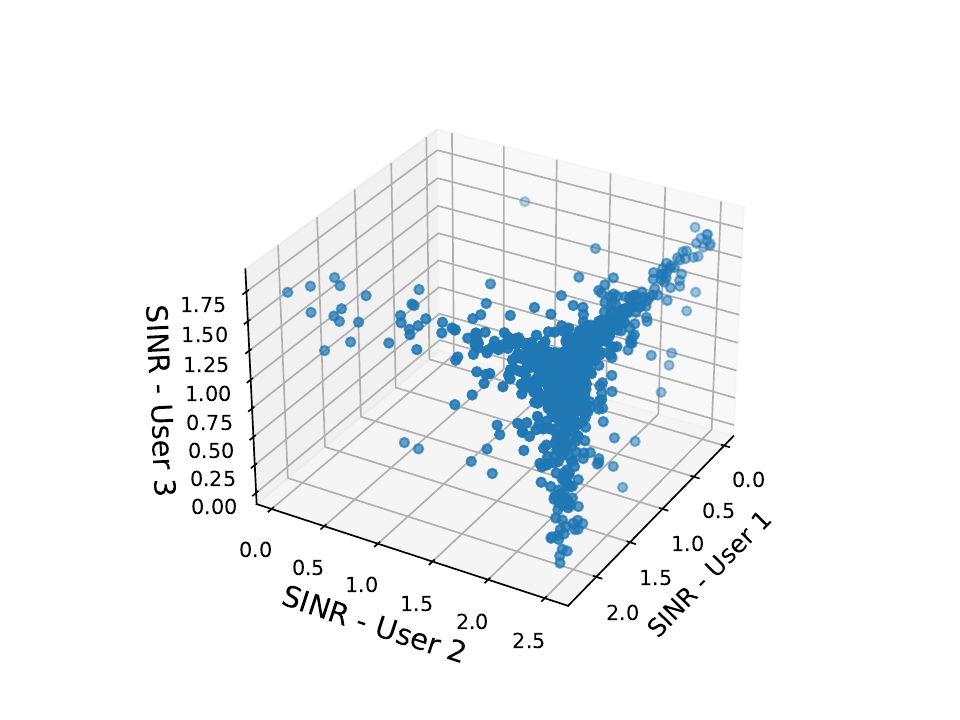} 
		\caption{SINR}
		\label{fig:subfig1}
	\end{subfigure}
	\hfill
	\begin{subfigure}[b]{0.45\columnwidth}
		\centering		
		\includegraphics[width=\textwidth, trim=2cm 0cm 2cm 2cm, clip]{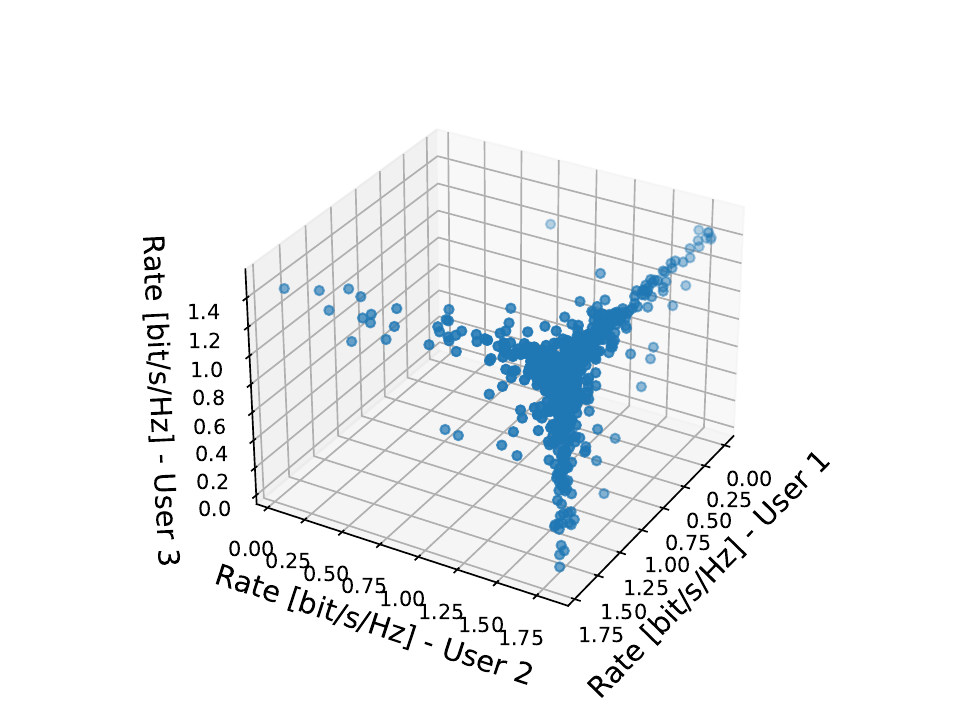} 
		\caption{Achievable rate}
		\label{fig:subfig2}
	\end{subfigure}
	\caption{\footnotesize Sample points on the weak Pareto boundary when the conditions of Corollary~\ref{cor.withpc} are satisfied.}
	\label{fig.convex}
	\vspace{-.5cm}
\end{figure}

Fig.~\ref{fig.nonconvex} presents the results of a user placement scenario where the matrices $(\signal{M}+\signal{u}\signal{a}^t_n)_{n\in\{1,2,3\}}$ are not inverse Z-matrices. To be precise, one of these matrices is given by 
\begin{align*}
	\begin{bmatrix}
		3.4 \times 10^{-1} & 1.4 \times 10^{-4} & 9.4 \times 10^{-2} \\
		5.8 \times 10^{-2} & 4.4 \times 10^{-1} & 4.3 \times 10^{-2} \\
		3.4  & 7.4 \times 10^{-4} & 5.0 \times 10^{-1}
	\end{bmatrix},
\end{align*}
which, in light of Remark~\ref{remark.zmatrix}, cannot be an inverse Z-matrix because, in particular, the determinant of the $2\times 2$ matrix obtained by removing the second row and the second column is negative; i.e., users 1 and 3 do not offer Z-compatible channels in this three-user network. Consequently, our theoretical results cannot guarantee convexity of the SINR and rate regions, and we indeed verify nonconvexity for the SINR region. However, Fig.~\ref{fig:subfig22} indicates that the rate region may still exhibit convexity. This result does not contradict Corollary~\ref{cor.withpc} because the conditions we derived are sufficient but not necessary.  Nevertheless, the potential convexity of the rate region, despite a nonconvex SINR region, provides numerical evidence supporting the design recommendation outlined in Remark~\ref{remark.design}.

\begin{figure}
	\centering
	\begin{subfigure}[b]{0.49\columnwidth}
		\centering
		\includegraphics[width=\textwidth, trim=2cm 0cm 2cm 2cm, clip]{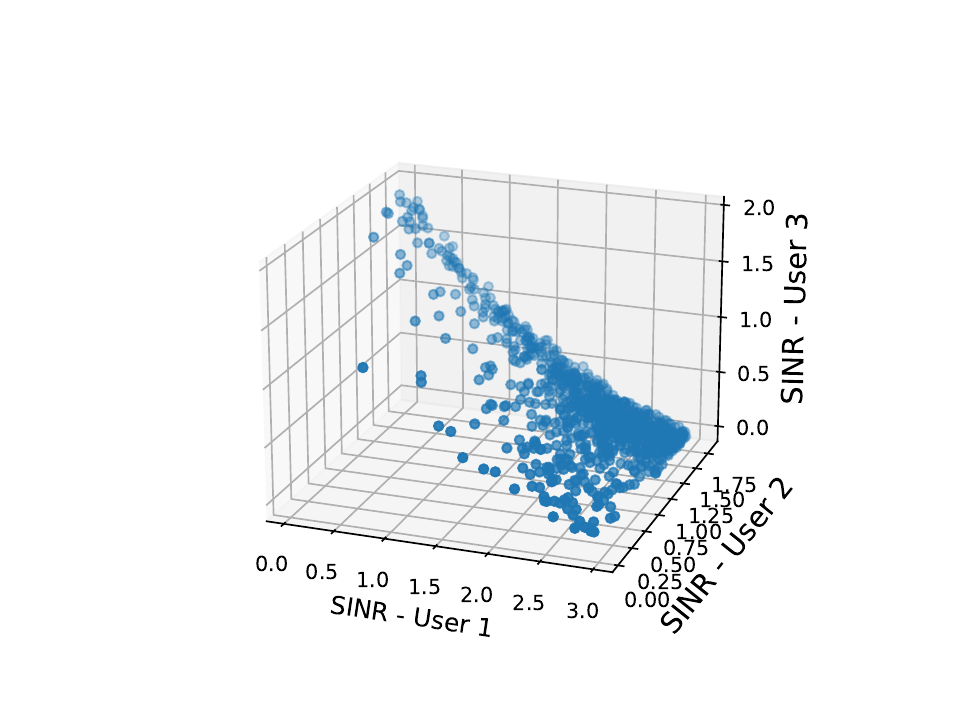} 
		\caption{SINR}
		\label{fig:subfig21}
	\end{subfigure}
	\hfill
	\begin{subfigure}[b]{0.49\columnwidth}
		\centering
		\includegraphics[width=\textwidth, trim=2cm 0cm 2cm 2cm, clip]{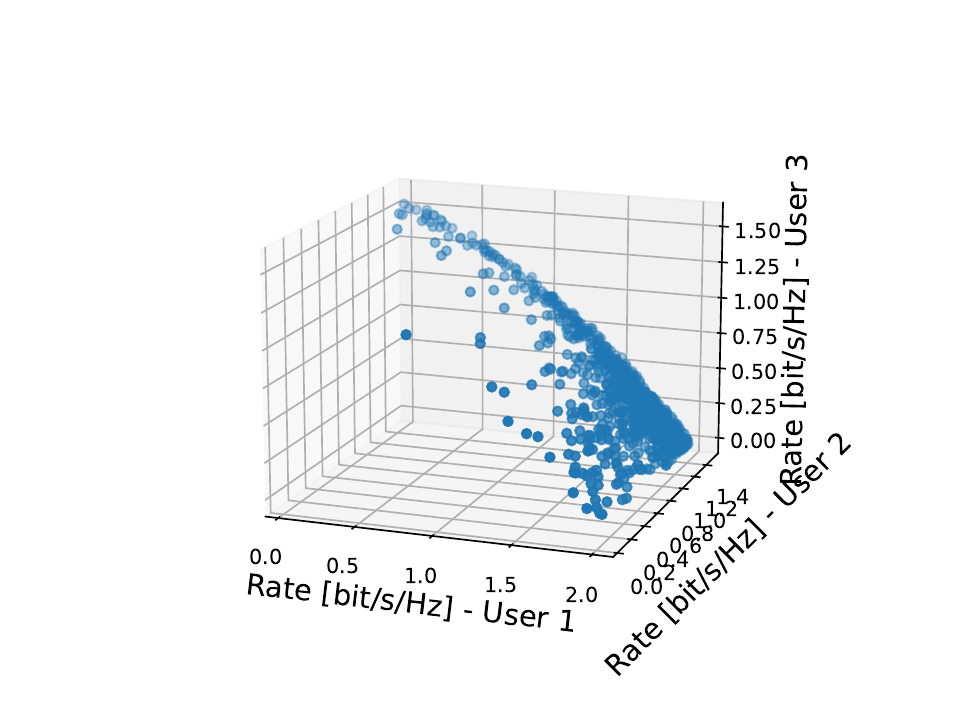} 
		\caption{Achievable rate}
		\label{fig:subfig22}
	\end{subfigure}
	\caption{\footnotesize Sample points on the weak Pareto boundary when the conditions of Corollary~\ref{cor.withpc} are not satisfied.}
	\label{fig.nonconvex}
	\vspace{-.3cm}
\end{figure}

\section{Summary}
We have derived simple characterizations of the SINR and rate regions, with and without power constraints, using the concept of spectral radius of nonlinear mappings. This viewpoint streamlines the analysis and connects the study of feasible regions in wireless networks to a well-developed mathematical framework, opening the door to further extensions. In particular, our characterization of the feasible rate region as the level-set of a function involving the nonlinear spectral radius enables  the identification of regimes in which time sharing brings no benefit with respect to the achievable rates. In particular, for affine interference models that explicitly account for self-interference, we show that convexity of the SINR and rate regions is guaranteed under a tractable sufficient condition: the relevant interference matrices are inverse Z-matrices. We further connect common interference patterns in wireless networks to the extensive theory of inverse Z-matrices, providing tools to recognize convex feasible regions without explicitly computing matrix inverses. Beyond these structural results, we disprove the conjecture posed in \cite[Sect.~5.4.4]{slawomir09}\cite{stanczak2007convexity}, and we introduce the notion of channel Z-compatibility, which overcomes some limitations of the concept of favorable propagation in the massive MIMO literature. Finally, supported by both theory and numerical evidence, we advocate formulating sum-rate maximization problems in terms of rates rather than SINR levels, as this choice can reveal hidden convexity and lead to tractable large-scale optimization problems.

\bibliographystyle{IEEEtran}
\bibliography{IEEEabrv,references}

\clearpage
\section*{Supplemental material}

Below, we collect well-known results from the literature that serve as key components in the proofs presented in our study.

\begin{fact}
	\label{fact.2x2}
	\cite[Theorem 2.9.2]{johnson2011} A nonnegative matrix $\signal{M}\in\real^{N\times N}_+$ is an inverse Z-matrix for $N=2$ if and only if $\mathrm{det}(\signal{M})>0$. 
\end{fact}

\begin{fact}
	\label{fact.diag}
	
	\cite[Theorems 4.5 and 9.2]{johnson2011} Let $\signal{M}\in\real_{++}^{N\times N}$. Then there exists $\alpha>0$ such that $\alpha\signal{I}+\signal{M}$ is an inverse Z-matrix.
\end{fact}

\begin{fact}
	\label{fact.dsum}
	\cite[Theorem~1.7]{johnson2011} If $\signal{M}\in\real^{N\times N}$ is an inverse M-matrix, then $\signal{M}+\signal{D}$ is also an inverse M-matrix for any nonnegative diagonal matrix $\signal{D}\in\real_{+}^{N\times N}$.
\end{fact}

\begin{fact}
	\label{fact.zm}
	\cite[Theorem 1.1]{johnson2011} Let the nonnegative matrix $\signal{M}\in\real_+^{N\times N}$ be invertible. Then $\signal{M}$ is an inverse Z-matrix if and only if $\signal{M}^{-1}$ is an M-matrix (i.e., $\signal{M}$ is an inverse M-matrix).
\end{fact}

\begin{fact} \cite[Theorem~4.3]{friedland81}
	\label{fact.friedland}
	Assume that the nonnegative matrix $\signal{M}\in\real_+^{N\times N}$ is an inverse Z-matrix. Then the mapping $\mathcal{D}_{+}^{N\times N}\to\real_+:\signal{D}\mapsto \rho(\signal{DM})$ is a convex function. 
\end{fact}

\begin{fact}
	\label{fact.cond_eig} \cite{nuzman07}
	Let $\|\cdot\|$ be a monotone norm. Assume that $T:\real_{+}^N\to\real_{++}^N$ is a standard interference mapping. Then each of the following holds:
\begin{itemize}
	\item[(i)] There exists a unique solution $(\signal{x}^\star, \lambda^\star)\in\real_{++}^N\times\real_{++}$ to the conditional eigenvalue problem
	\begin{problem}
		\label{problem.cond_eig}
		Find $(\signal{x}, \lambda)\in\real_{+}^N\times\real_{+}$ such that $T(\signal{x})=\lambda\signal{x}$ and $\|\signal{x}\|=1$.
	\end{problem}

	\item[(ii)] The sequence $(\signal{x}_n)_{n\in\Natural}\subset\real_{+}^N$ generated via 
	\begin{align}
		\label{eq.krause_iter}
		\signal{x}_{n+1} = T^\prime({\signal{x}_n}):=\dfrac{1}{\|T(\signal{x}_n)\|}T(\signal{x}_n),\quad\signal{x}_1\in\real_{++}^N,
	\end{align}
	converges to the uniquely existing vector $\signal{x}^\star\in\mathrm{Fix}(T^\prime):=\{\signal{x}\in\real_{+}^N~|~\signal{x}=T^\prime(\signal{x})\}$, which is also the vector $\signal{x}^\star$ of the tuple $(\signal{x}^\star,\lambda^\star)$ that solves Problem~\ref{problem.cond_eig}. Furthermore, the sequence $(\lambda_n:=\|T(\signal{x}_n)\|)_{n\in\Natural}\subset\real_{++}$ converges to $\lambda^\star$.
\end{itemize}
\end{fact}

The next result is immediate from \cite[Lemmas 3.2 and 3.3]{nussbaum1986convexity} and the definition of the nonlinear spectral radius.

\begin{fact}
	\label{fact.ineq_spec}
	Let $G:\real_+^N\to\real_+^N$ be a general interference mapping. Then each of the following holds:
	\item[(i)] if $(\signal{x},\lambda)\in\real_{++}^N\times\real_{++}$ and $G(\signal{x})\le \lambda \signal{x}$, then $\rho(G)\le\lambda$; 
	\item[(ii)] if $(\signal{x},\lambda)\in\real_{++}^N\times\real_{++}$ and $G(\signal{x}) = \lambda \signal{x}$, then $\rho(G) = \lambda$; and
	\item[(iii)] if $(\signal{x},\lambda)\in\real_{+}^N\backslash\{\signal{0}\}\times\real_{+}$ and $G(\signal{x}) \ge \lambda \signal{x}$ then $\rho(G)\ge\lambda$.
\end{fact}

\begin{fact}
	\label{fact.continuity} \cite[Theorem~3.1(2)]{nussbaum1986convexity} Fix a norm $\|\cdot\|$ in $\real^N$, and define $\mathcal{P}:=\{\signal{x}\in\real_+^N\mid \|\signal{x}\|=1\}$. For a sequence $(G_n:\real_+^N\to\real_+^N)_{n\in\Natural}$ of general interference mappings, assume that there exists a general interference mapping $G:\real_+^N\to\real_+^N$ such that $\lim_{n\to\infty}\sup_{\signal{x}\in\mathcal{P}}\|G_n(\signal{x})-G(\signal{x})\|=0$. Then $\lim_{n\to\infty}\rho(G_n)=\rho(G)$.
\end{fact}

\begin{fact}
	\label{fact.tinf}
	\cite[Proposition 4]{cavalcante2019} Let $T:\real_+^N\to\real_{++}^N$ be a standard interference mapping and $T_\infty:\real_+^N\to\real_{+}^N$ be its asymptotic mapping (see Definition~\ref{def.am}). Then $\mathrm{Fix}(T)\neq\emptyset$ if and only if the (nonlinear) spectral radius $\rho(T_\infty)$ of $T_\infty$ satisfies $\rho(T_\infty)<1$.
\end{fact}

\begin{fact}
	\label{fact.feasibility}
	\cite[Proposition 5]{cavalcante2019}~Let $T:\real^N_{+}\to\real_{++}^N$ be a standard interference mapping and $\|\cdot\|$ a monotone norm. Then $T$ has a fixed point $\signal{x}^\star$ satisfying $\|\signal{x}^\star\|\le 1$ if and only if the conditional eigenvalue $\lambda$ associated with $T$ and $\|\cdot\|$ satisfies $\lambda\le 1$.
\end{fact}

\end{document}

%% file: macros_reduced.tex
\usepackage{bm}

\renewcommand{\vec}[1]{\bm{#1}}		

\usepackage{amsthm}
\usepackage{amsmath}
\DeclareMathOperator*\argmax{arg \, max}		

\newtheorem{proposition}{Proposition}

\newtheorem{lemma}{Lemma}

\newtheorem{Cor}{Corollary}

\newtheorem{definition}{Definition}
\newtheorem{remark}{Remark}
\newtheorem{fact}{Fact}

\newtheorem{problem}{Problem}
\newtheorem{example}{Example}

\newcommand{\N}{{\mathbb N}}

\newcommand{\ngn}{g_{T_{\|\cdot\|}}}

\newcommand{\bsp}{\overline{{\mathcal{S}}_{\mathcal{P}}}}
\newcommand{\brp}{\overline{\mathcal{R}_{\mathcal{P}}}}
\newcommand{\diagvec}[1]{{\mathrm{diag}(\boldsymbol{#1})}}
\newcommand{\signal}[1]{{\boldsymbol{#1}}}

\newcommand{\real}{{\mathbb R}}

\newcommand{\Natural}{{\mathbb N}}
\newcommand{\refeq}[1]{(\ref{#1})}

\newcommand{\tn}[1]{{(t_{#1})_{\|\cdot\|}}}
\newcommand{\tnorm}{{T_{\|\cdot\|}}}

%% file: IEEEabrv.bib
@STRING{IEEE_J_SP         = "{IEEE} Trans. Signal Processing"}

@STRING{IEEE_J_JSAC       = "{IEEE} J. Select. Areas Commun."}


%% file: references.bib
@Article{krause01,
  Title                    = {Concave {Perron--Frobenius} theory and applications},
  Author                   = {Krause, Ulrich},
  Journal                  = {Nonlinear Analysis: Theory, Methods \& Applications},
  Year                     = {2001},
  Number                   = {3},
  Pages                    = {1457--1466},
  Volume                   = {47},
  Publisher                = {Elsevier}
}

@Article{cavalcante2016,
  Title                    = {Elementary Properties of Positive Concave Mappings with Applications to Network Planning and Optimization},
  Author                   = {Cavalcante, Renato L.~G.~ and Shen, Yuxiang and Sta{\'n}czak, Slawomir},
  Journal                  = IEEE_J_SP,
  Month                    = {April},
  Year					   = {2016},
  Pages                    = {1774-1873},
  Volume                   = {64},
  Number                   = {7}
}

@Article{yates95,
  Title                    = {A framework for uplink power control in cellular radio systems},
  Author                   = {R. D. Yates},
  Journal                  = IEEE_J_JSAC,
  Year                     = {1995},
  Month                    = {Sept.},
  Number                   = {7},
  Pages                    = {1341-1348},
  Volume                   = {13}
}

@Book{baus17,
	Title                    = {Convex Analysis and Monotone Operator Theory in Hilbert Spaces},
	Author                   = {H. H. Bauschke and P. L. Combettes},
	Publisher                = {Springer},
	Year                     = {2017},
    Edition                  = {2th edition}
}

@Book{boyd,
  Title                    = {Convex Optimization},
  Author                   = {S. Boyd and L. Vandenberghe},
  Publisher                = {Cambridge Univ. Press},
  Year                     = {2006},
  Address                  = {Cambridge, U.K.}
}

@Book{martin11,
  Title                    = {Interference Calculus - A General Framework for Interference Management and Network Utility Optimization},
  Author                   = {Martin Schubert and Hoger Boche},
  Publisher                = {Springer},
  Year                     = {2011},
  Address                  = {Berlin}
}

@Book{slawomir09,
  Title                    = {Fundamentals of Resource Allocation in Wireless Networks},
  Author                   = {S. Sta\'nczak and M. Wiczanowski and H. Boche},
  Editor                   = {W. Utschick and H. Boche and R. Mathar},
  Publisher                = {Springer},
  Year                     = {2009},
  Address                  = {Berlin Heidelberg},
  Edition                  = {2nd},
  Series                   = {Foundations in Signal Processing, Communications and Networking}
}

@book{rock70,
	title={Convex analysis},
	author={R.~ Tyrrell~Rockafellar},
	year={1970},
	publisher={Princeton university press}
}

@inproceedings{nuzman07,
	title={Contraction approach to power control, with non-monotonic applications},
	author={Nuzman, Carl J},
	booktitle={IEEE GLOBECOM 2007-IEEE Global Telecommunications Conference},
	pages={5283--5287},
	year={2007},
	organization={IEEE}
}

@inproceedings{renatomaxmin,
	title={Fundamental properties of solutions to utility maximization problems},
	author={R.~L.~.G.~Cavalcante and S.~Sta\'nczak},
	booktitle = {arXiv:1610.01988},
    year={2016}
}

@article{nussbaum1986convexity,
	title={Convexity and log convexity for the spectral radius},
	author={Nussbaum, Roger D},
	journal={Linear Algebra and its Applications},
	volume={73},
	pages={59--122},
	year={1986},
	publisher={Elsevier}
}

@book{tan2014wireless,
	title={Wireless Network Optimization by {Perron-Frobenius} Theory},
	author={Tan, Chee Wei},
	Publisher={Now Publishers Inc},
	Year={2015}
}

@book{marzetta16,
	title={Fundamentals of Massive {MIMO}},
	author={T.~L.~Marzetta and E.~G.~Larsson and H.~Yang and H.~Q.~Ngo},
	year={2016},
	publisher={Cambridge University Press}
}

@article{burbanks2003extension,
	title={Extension of order-preserving maps on a cone},
	author={Burbanks, Andrew D and Nussbaum, Roger D and Sparrow, Colin T},
	journal={Proceedings. Section A, Mathematics-The Royal Society of Edinburgh},
	volume={133},
	number={1},
	pages={35},
	year={2003},
	publisher={Cambridge University Press}
}

@article{oshime92,
	title={ {Perron--Frobenius} problem for weakly sublinear maps in a {Euclidean} positive orthant},
	author={Oshime, Yorimasa},
	journal={Japan journal of industrial and applied mathematics},
	volume={9},
	number={2},
	pages={313},
	year={1992},
	publisher={Springer}
}

@article{boche2008,
	title={Concave and convex interference functions -- General characterizations and applications},
	author={Boche, Holger and Schubert, Martin},
	journal={IEEE Transactions on Signal Processing},
	volume={56},
	number={10},
	pages={4951--4965},
	year={2008},
	publisher={IEEE}
}

@article{cavalcante2019,
	title={Connections between spectral properties of asymptotic mappings and solutions to wireless network problems},
	author={Cavalcante, Renato L. G. and Liao, Qi and Sta{\'n}czak, Slawomir},
	journal={IEEE Transactions on Signal Processing},
	volume={67},
	number={10},
	pages={2747--2760},
	year={2019},
	publisher={IEEE}
}

@article{friedland81,
	title={Convex spectral functions},
	author={Friedland, Shmuel},
	journal={Linear and multilinear algebra},
	volume={9},
	number={4},
	pages={299--316},
	year={1981},
	publisher={Taylor \& Francis}
}

@article{johnson2011,
	title={Inverse {M}-matrices, ii},
	author={Johnson, Charles R and Smith, Ronald L},
	journal={Linear algebra and its applications},
	volume={435},
	number={5},
	pages={953--983},
	year={2011},
	publisher={Elsevier}
}

@article{demir2021,
	title={Foundations of user-centric cell-free massive {MIMO}},
	author={Demir, {\"O}.~T. and Bj{\"o}rnson, E. and Sanguinetti, L.},
	journal={Foundations and Trends in Signal Processing},
	volume = {14},
	pages = {162--472},
	year={2021}
}

@article{massivemimobook,
	year = {2017},
	volume = {11},
	journal = {Foundations and Trends{\textregistered} in Signal Processing},
	title = {Massive {MIMO} Networks: {Spectral}, Energy, and Hardware Efficiency},
	doi = {10.1561/2000000093},
	issn = {1932-8346},
	number = {3-4},
	pages = {154-655},
	author = {E. Bj\"{o}rnson and J. Hoydis and L. Sanguinetti}
}

@article{miretti2024ul,
	title={{UL-DL} duality for cell-free massive {MIMO} with {per-AP} power and information constraints},
	author={L.~Miretti and R.~L.~G.~Cavalcante and E.~Bj{\"o}rnson and S{\l}awomir Sta{\'n}czak},
	journal={IEEE Transactions on Signal Processing},
	year={2024},
	publisher={IEEE}
}

@inproceedings{boche2002,
	title={A general duality theory for uplink and downlink beamforming},
	author={Boche, Holger and Schubert, Martin},
	booktitle={Proceedings IEEE 56th Vehicular Technology Conference},
	volume={1},
	pages={87--91},
	year={2002},
	organization={IEEE}
}

@article{piotrowski2022,
	title={The fixed point iteration of positive concave mappings converges geometrically if a fixed point exists: Implications to wireless systems},
	author={T.~Piotrowski and R.~L.~G.~Cavalcante},
	journal={IEEE Transactions on Signal Processing},
	volume={70},
	pages={4697--4710},
	year={2022},
	publisher={IEEE}
}

@Book{lem13,
	Title                    = {Nonlinear {Perron-Frobenius} theory},
	Author                   = {B. Lemmens and R. Nussbaum},
	Publisher                = {Cambridge University Press},
	Year                     = {2012},
	Address                  = {Cambridge, UK}
}

@inproceedings{cavalcante2023,
	title={Characterization of the weak {Pareto} boundary of resource allocation problems in wireless networks--Implications to cell-less systems},
	author={Renato L. G. Cavalcante and L. Miretti and S. Sta{\'n}czak},
	booktitle={IEEE International Conference on Communications (ICC)},
	pages={5010--5016},
	year={2023},
	organization={IEEE}
}

@article{stanczak2007convexity,
	title={On the convexity of feasible {QoS} regions},
	author={Stanczak, Sawomir and Boche, Holger},
	journal={IEEE Transactions on Information Theory},
	volume={53},
	number={2},
	pages={779--783},
	year={2007},
	publisher={IEEE}
}

@article{chen2018achievable,
	title={When is the achievable rate region convex in two-user massive {MIMO} systems?},
	author={Chen, Zheng and Bj{\"o}rnson, Emil and Larsson, Erik G},
	journal={IEEE Wireless Communications Letters},
	volume={7},
	number={5},
	pages={796--799},
	year={2018},
	publisher={IEEE}
}

@article{tan2011maximizing,
	title={Maximizing sum rate and minimizing {MSE} on multiuser downlink: Optimality, fast algorithms and equivalence via max-min {SINR}},
	author={Tan, Chee Wei and Chiang, Mung and Srikant, R},
	journal={IEEE Transactions on Signal Processing},
	volume={59},
	number={12},
	pages={6127--6143},
	year={2011},
	publisher={IEEE}
}

@article{tan2011nonnegative,
	title={Nonnegative matrix inequalities and their application to nonconvex power control optimization},
	author={Tan, Chee Wei and Friedland, Shmuel and Low, Steven},
	journal={SIAM Journal on Matrix Analysis and Applications},
	volume={32},
	number={3},
	pages={1030--1055},
	year={2011},
	publisher={SIAM}
}

@article{miretti2025two,
	title={Two-timescale joint power control and beamforming design with applications to cell-free massive {MIMO}},
	author={Miretti, Lorenzo and Cavalcante, Renato LG and Sta{\'n}czak, S{\l}awomir},
	journal={IEEE Transactions on Wireless Communications},
	year={2025},
	publisher={IEEE},
	note = {to appear}
}

@article{cheng2016optimal,
	title={Optimal pilot and payload power control in single-cell massive {MIMO} systems},
	author={Cheng, Hei Victor and Bj{\"o}rnson, Emil and Larsson, Erik G},
	journal={IEEE Transactions on Signal Processing},
	volume={65},
	number={9},
	pages={2363--2378},
	year={2016},
	publisher={IEEE}
}

@article{luo2008dynamic,
	title={Dynamic spectrum management: Complexity and duality},
	author={Luo, Zhi-Quan and Zhang, Shuzhong},
	journal={IEEE journal of selected topics in signal processing},
	volume={2},
	number={1},
	pages={57--73},
	year={2008},
	publisher={IEEE}
}

@inproceedings{miretti2022closed,
	title={Closed-form max-min power control for some cellular and cell-free massive {MIMO} networks},
	author={Miretti, Lorenzo and Cavalcante, Renato LG and Sta{\'n}czak, S{\l}awomir and Schubert, Martin and B{\"o}hnke, Ronald and Xu, Wen},
	booktitle={2022 IEEE 95th Vehicular Technology Conference:(VTC2022-Spring)},
	pages={1--7},
	year={2022},
	organization={IEEE}
}

@article{friedland2008maximizing,
	title={Maximizing sum rates in {Gaussian} interference-limited channels},
	author={Friedland, Shmuel and Tan, Chee Wei},
	journal={arXiv preprint arXiv:0806.2860},
	year={2008}
}

@inproceedings{miretti2024sum,
	title={Two-timescale weighted sum-rate maximization for large cellular and cell-free massive {MIMO}},
	author={Miretti, Lorenzo and Bj{\"o}rnson, Emil and Sta{\'n}czak, S{\l}awomir},
	booktitle={2024 IEEE 25th International Workshop on Signal Processing Advances in Wireless Communications (SPAWC)},
	pages={656--660},
	year={2024},
	organization={IEEE}
}

@article{hanly1995algorithm,
	title={An algorithm for combined cell-site selection and power control to maximize cellular spread spectrum capacity},
	author={Hanly, Stephen V.},
	journal={IEEE Journal on selected areas in communications},
	volume={13},
	number={7},
	pages={1332--1340},
	year={1995},
	publisher={IEEE}
}

@inproceedings{tsiaflakis2008optimality,
	title={Optimality certificate of dynamic spectrum management in multi-carrier interference channels},
	author={Tsiaflakis, Paschalis and Tan, Chee Wei and Yi, Yung and Chiang, Mung and Moonen, Marc},
	booktitle={2008 IEEE International Symposium on Information Theory},
	pages={1298--1302},
	year={2008},
	organization={IEEE}
}

@article{chafaa2025,
	title={Transformer-Based Power Optimization for Max-Min Fairness in Cell-Free Massive {MIMO}},
	author={Chafaa, Irched and Bacci, Giacomo and Sanguinetti, Luca},
	journal={IEEE Wireless Communications Letters},
	year={2025},
	publisher={IEEE},
	pages = {2316--2320},
	month = {Aug.},
	volume = {14},
	number = {8}
}

@book{zalinescu2002,
	title={Convex analysis in general vector spaces},
	author={Zalinescu, Constantin},
	year={2002},
	publisher={World scientific}
}

@article{zheng2013maximizing,
	title={Maximizing sum rates in cognitive radio networks: Convex relaxation and global optimization algorithms},
	author={Zheng, Liang and Tan, Chee Wei},
	journal={IEEE Journal on Selected Areas in Communications},
	volume={32},
	number={3},
	pages={667--680},
	year={2013},
	publisher={IEEE}
}

@article{boche2010unifying,
	title={A unifying approach to interference modeling for wireless networks},
	author={Boche, Holger and Schubert, Martin},
	journal={IEEE Transactions on Signal Processing},
	volume={58},
	number={6},
	pages={3282--3297},
	year={2010},
	publisher={IEEE}
}

@article{weeraddana2012,
	title={Weighted sum-rate maximization in wireless networks: A review},
	author={Weeraddana, Pradeep Chathuranga and Codreanu, Marian and Latva-aho, Matti and Ephremides, Anthony and Fischione, Carlo and others},
	journal={Foundations and Trends{\textregistered} in Networking},
	volume={6},
	number={1--2},
	pages={1--163},
	year={2012},
	publisher={Now Publishers, Inc.}
}

@inproceedings{schubert2024duality,
	title={Duality-Based Joint Clustering and Precoding for Cell-Free Distributed MIMO},
	author={Schubert, Martin and B{\"o}hnke, Ronald and Xu, Wen},
	booktitle={2024 27th International Workshop on Smart Antennas (WSA)},
	pages={1--7},
	year={2024},
	organization={IEEE}
}
